\pdfoutput=1
\documentclass[11pt]{amsart}

\usepackage{amsmath}
  \usepackage{paralist}
  \usepackage{graphics} 
  \usepackage{epsfig} 
\usepackage{graphicx}  \usepackage{epstopdf}
 \usepackage[colorlinks=true]{hyperref}

\hypersetup{urlcolor=blue, citecolor=red}

\newtheorem{theorem}{Theorem}[section]

\newtheorem{lemma}[theorem]{Lemma}

\theoremstyle{definition}
\newtheorem{definition}[theorem]{Definition}

\title[Three Weight Ternary Linear Codes] 
      {Three Weight Ternary Linear Codes from non-weakly regular bent functions}

\author[Rum\.{ı} Mel\.{ı}h Pelen]{}



 \email{rumimelih.pelen@erzurum.edu.tr}

\newcommand{\C}{{\mathbb C}}

\newcommand{\F}{{\mathbb F}}

\newcommand{\Q}{{\mathbb Q}}

\newcommand{\Z}{{\mathbb Z}}

\newcommand{\sC}{{\mathcal C}}

\newcommand{\be}{\begin{eqnarray}}
\newcommand{\ee}{\end{eqnarray}}

\newcommand{\Tr}{{\rm Tr}}

\newcommand{\Dim}{{\rm Dim}}

\newcommand{\supp}{{\rm supp}}

\newcommand{\boxtensor}{{\Box\kern-9.03pt\raise1.42pt\hbox{$\times$}}}

\newcommand{\Image}{{\rm Im}}
\newcommand{\Kernel}{{\rm Ker \ }}

\newcommand\norm[1]{\left\lVert#1\right\rVert}
\begin{document}
	\newtheorem{teo}{Theorem}[section]
	\newtheorem{cas}{Case}[section]
	\newtheorem{observation}{Observation}
	\newtheorem{cor}{Corallary}[section]
	\newtheorem{lem}{Lemma}[section]
	\newtheorem{prp}{Proposition}[section]
	\newtheorem{rem}{Remark}[section]
	\newtheorem{cla}{Claim}[section]
	\newtheorem{df}{Definition}[section]
	\newtheorem{con}{Condition}
	\newtheorem*{case}{Case}
           \newtheorem{obsrv}{Observation}
	\newtheorem{fact}{Fact}
	\newtheorem{example}{Example}

	\newenvironment{pf}{\noindent\textbf{Proof.}\quad}{\hfill{$\Box$}}
	
\maketitle

\centerline{\scshape Rumi Melih Pelen}
\medskip
{\footnotesize

 \centerline{ Erzurum Technical University}
   \centerline{Yakutiye, Erzurum, Turkey}

\bigskip

\begin{abstract}
 In this paper, several classes of three-weight ternary linear codes from non-weakly regular dual-bent functions are constructed based on a generic construction method. Instead of the whole space, we use the subspaces $B_{\pm}(f)$ associated with a ternary non-weakly regular dual-bent function $f$. Unusually, we use the pre-image sets of the dual function $f^*$ in $B_{\pm}(f)$ as the defining sets of the corresponding codes. Since, the size of the defining sets of the constructed codes are flexible, it enables us to construct several codes with different parameters for a fixed dimension. We represent the weight distribution of the constructed codes. We also give several examples.
\end{abstract}

\section{Introduction} \label{intro}

Linear codes with a few weights have also applications in secret sharing \cite{anderson1998build, carlet2005linear, ding2015class, yuan2006secret}, authentication codes \cite{ding}, association schemes \cite{calder2}, and strongly regular graphs \cite{calder1}. Some interesting two-weight and three-weight codes can be found in \cite{ding2015class, mesnager2017linear, tang, zhou2016linear}. There are several  methods to construct linear codes, one of which is based on pre-image sets of functions over finite fields. Two generic constructions, which are called the first and second generic constructions, of linear codes from functions have been distinguished from the others in the literature. In \cite{ding2015class}, Ding and Ding constructed two or three weight linear codes from the \textit{planar} map $F(x)=x^2$ over finite fields of odd characteristic based on the second generic construction. Shortly after, Zhou et al. \cite{zhou2016linear} generalized this result to quadratic bent functions. After then, Tang et al. \cite{tang} has generalized these results by constructing several classes of two or three-weight linear codes from weakly regular bent functions over finite fields based on the second generic construction. Within this framework, we aim to construct ternary linear codes from non-weakly regular dual-bent functions based on the second generic construction. To do this, instead of the whole space we use the subsets $B_{\pm}(f)$ (see, \cite{op,op1}) associated with a ternary non-weakly regular dual-bent function $f$. We also note that the construction of linear codes from non-weakly regular bent functions over finite fields was studied by the author in his PhD thesis for the first time in the literature as far as we know \cite{pelen}.

\par The paper is organized as follows.
Section \ref{matback} introduces the main notations and provides the necessary mathematical background that will be needed in subsequent sections.
In Sections \ref{SectionConstruction} and \ref{SectionConstruction-}, we construct the classes of three-weight ternary linear codes on $B_{+}(f)$ and $B_{-}(f)$ from ternary non-weakly regular dual-bent functions based on the second generic construction. We represent the weight distribution of the constructed codes. In addition to these, we give several examples. We conclude in Section \ref{conc}.

\section{Preliminaries}
\label{matback}
Let $p$ be an odd prime and $\mathbb{F}_{p^n}$  the finite field of order $p^n.$ Since it is a vector  space of dimension $n$ over $\mathbb{F}_p$, we also use the notation $\mathbb{F}_{p}^n$ which consists of $n$-tuples of the prime field $\mathbb{F}_p.$  Let $f$ be  a function from $\mathbb{F}_{p}^n$ to $\mathbb{F}_p.$ The Walsh transform of $f$ at $\alpha\in \mathbb{F}_{p}^n$ is defined as a complex-valued function $\hat{f}$ on $\mathbb{F}_{p}^n,$
$$\hat{f}(\alpha)=\sum_{x \in \mathbb{F}_{p}^n} \epsilon_{p}^{f(x)-\alpha \cdot x}$$ where $ \epsilon_{p}=e^{\frac{2\pi i}{p}}$ and $\alpha \cdot x$ denotes the standart dot product in $\mathbb{F}_{p}^n$. The trace of $\alpha\in\mathbb {F}_{p^n}$ over $\mathbb {F}_{p}$ is defined as
$\Tr_{p}(\alpha)=\alpha+\alpha^{p}+\alpha^{p^{2}}+\cdots+\alpha^{p^{n-1}}$.  

Throughout the paper we fix the following notations.
\begin{itemize}
	\item $\Q$: Field of rational numbers,
	\item $\Z$: Ring of integers in $\Q$,
	\item $\Z^+$: Positive integers in $\Z$,
	\item $\C$: Field of complex numbers,
	\item $\norm{z}$: Magnitude of a complex number $z$,
	\item $D^{\star}$: The non-zero elements in a set $D$,
	\item $|D|$: The cardinality of a set $D$,
	\item $-A$: For any subset $A$ of an additive group it is the subset $\{-x: x \in A \},$
	\item $\lfloor\; \rfloor$: Floor function,
	\item $\perp$: Orthogonal complement,
	\item $A^C$: Complement of a set $A$ in $F_p^n$, i.e. $\{x \in \F_p^n: x \notin A \}.$
\end{itemize}
\subsection{Bent Functions}
The function $f$ is called \textit{bent} function if  $\norm{\hat{f}(\alpha)}=p^{n/2}$ for all $\alpha \in \mathbb{F}_{p}^n.$ The
normalized Walsh coefficient of a bent function $f$ at $\alpha$ is defined as $p^{-n/2}\hat{f}(\alpha).$ The normalized Walsh coefficients of a bent function $f$ are characterized in \cite{ksw} as follows.
$$p^{-n/2}\hat{f}(\alpha)= \left\{ \begin{array}{ll}
\pm \epsilon_{p}^{f^*(\alpha)} & \mbox{if $p^n\equiv 1$ mod $ 4$}; \\
\pm i \epsilon_{p}^{f^*(\alpha)} & \mbox{if $p^n\equiv 3$  mod  $4$},
\end{array}
\right.$$ where $f^*$ is a function from $\mathbb{F}_{p}^n $ to $\mathbb{F}_{p}$, which is called the dual of $f$. A bent function $f$ whose dual function $f^*$ is also bent is called \textit{dual-bent} function.

A bent function $f:\mathbb{F}_{p}^n\rightarrow\mathbb{F}_{p}$ is called \textit{regular} if for all $\alpha \in \mathbb{F}_{p}^n$, we have $$p^{-n/2}\hat{f}(\alpha)= \epsilon_{p}^{f^*(\alpha)}$$ and \textit{weakly regular} if for all $\alpha \in \mathbb{F}_{p}^n$,  $$p^{-n/2}\hat{f}(\alpha)= \xi \epsilon_{p}^{f^*(\alpha)},$$ where $\xi \in \{\pm 1, \pm i\}$ is independent from $\alpha$, otherwise it is called \textit{non-weakly regular}. 

\begin{rem}
	Dual-bent functions contain the class of weakly regular bent functions. On the other hand, non-weakly regular bent functions can be divided into two subclass as dual-bent and non dual-bent. 
\end{rem}

For an arbitrary function $f:\mathbb{F}_{p}^{n}\rightarrow\mathbb{F}_{p}$ and $i \in \F_p$, let $ N_i(f)$ denotes the cardinality $|\{x:x \in \mathbb{F}_{p}^{n}| f(x)=i\}|$.

\begin{prp} \label{f*0}
	Let $f:\mathbb{F}_{p}^{n}\rightarrow\mathbb{F}_{p}$  be a bent function such that $f(0)=j$ and $f(x)=f(-x).$ Then $f^{*}(x)=f^{*}(-x)$ and $f^{*}(0)=j.$
\end{prp}

\begin{proof}
	For all $\alpha \in \mathbb{F}_{p}^{n}$, we have $$\begin{array}{lll}
	\hat{f}(-\alpha)&=& \xi_{-\alpha}p^{\frac{n}{2}} \epsilon_{p}^{f^*(-\alpha)}=\sum_{x \in \mathbb{F}_{p}^n} \epsilon_{p}^{f(x)+\alpha \cdot x}=\sum_{x \in \mathbb{F}_{p}^n} \epsilon_{p}^{f(-x)-\alpha \cdot (-x)}= \hat{f}(\alpha)
	\\
	&=& \xi_{\alpha}p^{\frac{n}{2}} \epsilon_{p}^{f^*(\alpha)}.
	\end{array}$$ Hence, we prove $f^{*}(x)=f^{*}(-x).$ Put $f^{*}(0)=i_0.$  If $n$ is odd (resp. even), by \cite[Propositions $3.1$(resp. $3.2$)]{op}, we have $N_{i_0}(f)$ is an odd integer. Since $f(x)=f(-x)$, it is possible if and only if $i_0=j.$
\end{proof}

 For any bent function $f:\mathbb{F}_{p}^{n}\rightarrow\mathbb{F}_{p}$, let $B_{+}(f)$ and $B_{-}(f)$ be the partitions of $\F_{p}^{n}$  given by $$ B_{+}(f):=\{w: w \in \mathbb{F}_{p}^{n} \mid  \hat{f}(w)=\xi p^{\frac{n}{2}}\epsilon_{p}^{f^*(w)}\},$$ $$ B_{-}(f):=\{w: w \in \mathbb{F}_{p}^{n} \mid  \hat{f}(w)=-\xi p^{\frac{n}{2}}\epsilon_{p}^{f^*(w)}\},$$ where $\xi =1$ if  $p^n\equiv 1$  (mod $4$) and  $\xi=i$ if  $p^n\equiv 3$ (mod $4$).

 The following is given in \cite{op}. Any bent function $f:\mathbb{F}_{p}^{n}\rightarrow\mathbb{F}_{p}$\; is of two types. 

$$ \text{Type}\;(+)\;\; \text{if}\;\;\hat{f}(0)=\epsilon p^{\frac{n}{2}}\epsilon_{p}^{f^*(0)}, \; \epsilon \in \{1,i\}, $$ 

$$ \text{Type}\;(-) \;\; \text{if}\;\; \hat{f}(0)=\epsilon p^{\frac{n}{2}}\epsilon_{p}^{f^*(0)},  \; \epsilon \in \{-1,-i\}. $$ 

\begin{rem}\label{rem0}
	It is known that weakly regular bent functions appear in pairs and given a weakly regular bent function $f:\mathbb{F}_{p}^{n}\rightarrow\mathbb{F}_{p}$ we have (see \cite{hel}) \be \widehat{f^*}(\alpha)=\xi^{-1}p^{\frac{n}{2}}\epsilon_p^{f(-\alpha)}, \label{wrdual} \ee where $\widehat{f}(\alpha)=\xi p^{\frac{n}{2}}\epsilon_p^{f^*(\alpha)}$. It is easy to see that for $p^n\equiv 1$  (mod $4$) the types of $f$ and $f^*$ are same, for $p^n\equiv 3$ (mod $4$) they are of different types. Moreover, for any weakly regular bent function $f:\mathbb{F}_{p}^{n}\rightarrow\mathbb{F}_{p}$, if $f$ is of type $(+)$(resp. type $(-)$) then by definition we have $\hat{f}(0)=\xi p^{\frac{n}{2}}\epsilon_{p}^{f^*(0)},\text{where}\;\xi \in \{1,i\}(\text{resp.} \{-1,-i\}) $ which implies $0 \in B_{+}(f)(\text{resp.}\; B_{-}(f))$. Since $f$ is a weakly regular bent function then we know that $\xi$ is independent from $\alpha \in \mathbb{F}_{p}^n$, where $p^{-n/2}\hat{f}(\alpha)= \xi \epsilon_{p}^{f^*(\alpha)}$. Hence, we have $\alpha \in B_{+}(f)(\text{resp.}\; B_{-}(f))$  for  all $\alpha \in \mathbb{F}_{p}^n$. Therefore, $B_{\pm}(f)=\mathbb{F}_{p}^{n}$ and $B_{\mp}(f)=\emptyset$ respectively. On the other hand if $f$ is non-weakly regular bent then we have $B_{\pm}(f)\ne \emptyset$.
\end{rem}

\begin{rem} \label{type}
	Observe that for any bent function $f:\mathbb{F}_{p}^{n}\rightarrow\mathbb{F}_{p},\;\alpha \in B_{\pm}(f)$ if and only if $ f(x)+\alpha\cdot x$ is of type\;$(\pm)$ respectively.
\end{rem}

In the remaining part of the paper we state $f$ is always a non-weakly regular bent function. 

For $f:\mathbb{F}_{p}^{n}\rightarrow\mathbb{F}_{p}$, by the inverse Walsh transform we have \be p^{n}\epsilon_{p}^{f(y)}= \sum_{\alpha \in \mathbb{F}_{p}^{n}} \epsilon_{p}^{\alpha \cdot y}\hat{f}(\alpha). \label{invwalsh}\ee

As $f$ is non-weakly regular bent, for $\alpha \in \mathbb{F}_{p}^{n},\; \text{we have}\;\;\hat{f}(\alpha)= \xi_{\alpha}p^{\frac{n}{2}} \epsilon_{p}^{f^*(\alpha)},$ where $\xi_{\alpha} \in \{\pm 1\}$ if  $p^n\equiv 1$  (mod $4$) and  $\xi_{\alpha} \in \{\pm i\}$ if  $p^n\equiv 3$ (mod $4$). Then by Equation \eqref{invwalsh} we have
	
	$$\begin{array}{lll}
	
	p^{\frac{n}{2}}\epsilon_{p}^{f(y)}&=&\sum_{\alpha \in \mathbb{F}_{p}^{n}} \xi_{\alpha} \epsilon_{p}^{f^*(\alpha)+\alpha \cdot y}
	\\
	&=& \xi \big{(} \sum_{\alpha \in B_{+}(f)} \epsilon_{p}^{f^*(\alpha)+\alpha \cdot y}- \sum_{\alpha \in B_{-}(f)} \epsilon_{p}^{f^*(\alpha)+\alpha \cdot y} \big{)},
	\end{array} $$
	
where $\xi = 1$ if  $p^n\equiv 1$  (mod $4$) and  $\xi = i$ if  $p^n\equiv 3$ (mod $4$).
Hence, we arrive at
\be \xi^{-1}p^{\frac{n}{2}}\epsilon_{p}^{f(y)}=\sum_{\alpha \in B_{+}(f)} \epsilon_{p}^{f^*(\alpha)+\alpha \cdot y}- \sum_{\alpha \in B_{-}(f)} \epsilon_{p}^{f^*(\alpha)+\alpha \cdot y}. \label{invwlsh}\ee

For any $y \in \mathbb{F}_p^{n}$ the complex numbers $S_{0}(f,y),\; S_{1}(f,y)$ are  defined in \cite{op} as follows.

$$  S_{0}(f,y)=\sum_{\alpha \in B_{+}(f)} \epsilon_{p}^{f^*(\alpha)+\alpha \cdot y}, \;\;\;\;\;\; S_{1}(f,y)=\sum_{\alpha \in B_{-}(f)} \epsilon_{p}^{f^*(\alpha)+\alpha \cdot y}. $$

Hence, by Equation \eqref{invwlsh}, we have
\be \xi^{-1}p^{\frac{n}{2}}\epsilon_{p}^{f(y)}= S_{0}(f,y)-S_{1}(f,y). \label{inverse}\ee
Hence, we arrive at 
\be S_{0}(f,y)-S_{1}(f,y)= \left\{ \begin{array}{ll} 
p^{\frac{n}{2}}\epsilon_{p}^{f(y)} & \mbox{if $n$ even or $n$ odd and $p\equiv 1$ mod $ 4$};\\
-ip^{\frac{n}{2}}\epsilon_{p}^{f(y)}& \mbox{if $n$ odd and $p\equiv 3$  mod  $4$}.
\end{array}
\right. \label{inversee}\ee
\begin{definition}\label{preimages}
	Let $f:\mathbb{F}_{p}^{n}\rightarrow\mathbb{F}_{p}$ be a non-weakly regular bent function, for any $i\in\mathbb{F}_{p}$ we further define the sets $C_i(f),\;D_i(f)$ as follows:  $$ C_i(f):=\{\alpha: \alpha \in  B_{+}(f)\mid f^*(\alpha)=i\},\; D_i(f):=\{\alpha: \alpha \in  B_{-}(f)\mid f^*(\alpha)=i\}.$$   

\end{definition}
From Definition \ref{preimages} we have \be \label{preim} B_{+}(f)=\cup_{i=0}^{p-1}C_i(f),\;\;B_{-}(f)=\cup_{i=0}^{p-1}D_i(f) \ee
\begin{prp}\label{dual}
	Let $f:\mathbb{F}_{p}^{n}\rightarrow\mathbb{F}_{p}$  be a non-weakly regular dual-bent function. If $n$ even or $n$ odd and $p\equiv 1$ mod $4$,  then we have $$ S_{0}(f,y)=\left\{ \begin{array}{ll} p^{\frac{n}{2}}\epsilon_{p}^{f(y)}  & \mbox{if  $f^*(x)+x \cdot y$ is of type\;$(+)$}; \\
	0 & \mbox{if $f^*(x)+x \cdot y$ is of type\;$(-)$};
	\end{array}
	\right.$$
	$$
	S_{1}(f,y)=\left\{ \begin{array}{ll} -p^{\frac{n}{2}}\epsilon_{p}^{f(y)}  & \mbox{if  $f^*(x)+x \cdot y$ is of type\;$(-)$}; \\
	0 & \mbox{if $f^*(x)+x \cdot y$ is of type\;$(+)$};
	\end{array}
	\right.
	$$
	
	If $n$ odd and $p\equiv 3$  mod  $4$, then we have
	$$ S_{0}(f,y)=\left\{ \begin{array}{ll} -ip^{\frac{n-1}{2}}\sqrt{p}\epsilon_{p}^{f(y)}  & \mbox{if  $f^*(x)+x \cdot y$ is of type\;$(-)$}; \\
	0 & \mbox{if  $f^*(x)+x \cdot y$ is of type\;$(+)$};
	\end{array}
	\right.$$
	$$
	S_{1}(f,y)=\left\{ \begin{array}{ll} ip^{\frac{n-1}{2}}\sqrt{p}\epsilon_{p}^{f(y)}  & \mbox{if  $f^*(x)+x \cdot y$ is of type\;$(+)$}; \\
	0 & \mbox{if  $f^*(x)+x \cdot y$ is of type\;$(-)$}.
	\end{array}
	\right.
	$$ 
\end{prp}

\begin{proof}
	The proof follows from \cite[Lemmas 3.4 and 3.5]{op}.
\end{proof}

\begin{obsrv} \label{obs}
Let $f:\mathbb{F}_{p}^{n}\rightarrow\mathbb{F}_{p}$  be a bent function such that $f(x)=f(-x)$. Then, we have $\hat{f}(-\alpha)=\sum_{x \in \mathbb{F}_{p}^n} \epsilon_{p}^{f(x)+\alpha \cdot x}= \sum_{x \in \mathbb{F}_{p}^n} \epsilon_{p}^{f(-x)-(\alpha \cdot -x)}=\hat{f}(\alpha)$. So, we deduce that if  $\alpha \in  B_{+}(f)$ (resp. $B_{-}(f)$) then $-\alpha \in  B_{+}(f)$ (resp. $B_{-}(f)$). Hence, we have $ B_{\pm}(f)=-B_{\pm}(f)$.
\end{obsrv}

\begin{prp}\label{dualtype}
Let $f:\mathbb{F}_{p}^{n}\rightarrow\mathbb{F}_{p}$  be a dual-bent function such that $f(x)=f(-x)$. Then, the types of $f$ and $f^*$ are same for $p^n\equiv 1$  mod  $4$, and different for  $p^n\equiv 3$  mod  $4$.
\end{prp}

\begin{proof}
Let $f$ is of type\;$(+)$.  Put $f(0)=j_0$. Since $f$ is of type\;$(+)$ then $0 \in B_{+}(f)$. From Observation \ref{obs}, we have also $B_{+}(f)=-B_{+}(f)$. So, we can see that $|B_{+}(f)|$ is an odd integer. Then, the equality $|B_{+}(f)|+|B_{-}(f)|=p^n$ implies that $|B_{-}(f)|$ is an even integer.

\begin{case} $n$ is even

 On the contrary assume that $f^*$ is of type\;$(-)$.  Then, by Proposition \ref{dual} we have $$-p^{\frac{n}{2}}\epsilon_{p}^{f(0)}=\sum_{\alpha \in B_{-}(f)} \epsilon_{p}^{f^*(\alpha)}.$$ Then we have,  $$(|D_{j_0}(f)|+p^{\frac{n}{2}})\epsilon_{p}^{j_0}+\sum_{j\neq j_0}|D_{j}(f)|\epsilon_{p}^{j}=0.$$ By the orthogonality relations of character sums there exists an integer $k$ such that $|D_{j_0}(f)|=k-p^{\frac{n}{2}}$ and $|D_{j}(f)|=k$ for all $j\neq j_0$. Hence, we have $$k=\frac{|B_{-}(f)|+p^{\frac{n}{2}}}{p}.$$ The fact that $|B_{-}(f)|$ is an even integer implies $k$ is an odd integer. On the other hand, from Observation \ref{obs} respectively from Proposition \ref{f*0} we have $B_{-}(f)=-B_{-}(f)$ respectively $f^{*}(x)=f^{*}(-x)$. These imply that $D_{j}(f)=-D_{j}(f)$ for all $j \in \F_p$. Hence, we deduce that $k$ is an even integer which gives a contradiction

\end{case}

\begin{case} $n$ is odd and $p \equiv 1$ mod  $4$.

On the contrary assume that $f^*$ is of type\;$(-)$. Then by Equation \eqref{gauss} and Proposition \ref{dual} we have $$-p^{\frac{n}{2}}\epsilon_{p}^{f(0)}=-p^{\frac{n-1}{2}}\sum_{j\in\mathbb{F}_p^{\star}}\left( \frac{j}{p} \right)\epsilon_p^{j_0+j}=\sum_{\alpha \in B_{-}(f)} \epsilon_{p}^{f^*(\alpha)}.$$ Then we have,  $$|D_{j_0}(f)|\epsilon_{p}^{j_0}+\sum_{j\neq 0}(|D_{j}(f)|+\left( \frac{j}{p} \right)p^{\frac{n-1}{2}})\epsilon_{p}^{j_0+j}=0.$$ By the orthogonality relations of character sums there exists an integer $k$ such that $|D_{j_0}(f)|=k$, $|D_{j+j_0}(f)|=k-\left( \frac{j}{p} \right)p^{\frac{n-1}{2}}$ for all $j\neq 0$. Hence, we have $$k=\frac{|B_{-}(f)|}{p}.$$ The fact that $|B_{-}(f)|$ is an even integer implies $k$ is an even integer. On the other hand, we have $B_{-}(f)=-B_{-}(f)$ and $f^{*}(x)=f^{*}(-x)$. This implies that $D_{j}(f)=-D_{j}(f)$  and $|D_{j}(f)|$ is even for all $j \in \F_p$. For $j\ne j_0$ we have $|D_{j}(f)|=k-\left( \frac{j-j_0}{p} \right)p^{\frac{n-1}{2}}$. Combining this with the observation $|D_{j}(f)|$ is even for all $j \in \F_p$  we deduce that $k$ is an odd integer. Hence, we arrive at a contradiction.

\end{case}
\begin{case} $n$ is odd and $p \equiv 3$ mod  $4$.

On the contrary assume that $f^*$ is of type\;$(+)$. Then by Equation \eqref{gauss} and Proposition \ref{dual} we have $$p^{\frac{n}{2}}\epsilon_{p}^{f(0)}=p^{\frac{n-1}{2}}\sum_{j\in\mathbb{F}_p^{\star}}\left( \frac{j}{p} \right)\epsilon_p^{j_0+j}=\sum_{\alpha \in B_{-}(f)} \epsilon_{p}^{f^*(\alpha)}.$$ Then we have,  $$|D_{j_0}(f)|\epsilon_{p}^{j_0}+\sum_{j\neq 0}(|D_{j}(f)|-\left( \frac{j}{p} \right)p^{\frac{n-1}{2}})\epsilon_{p}^{j_0+j}=0.$$ By the orthogonality relations of character sums there exists an integer $k$ such that $|D_{j_0}(f)|=k$ and $|D_{j+j_0}(f)|=k+\left( \frac{j}{p} \right)p^{\frac{n-1}{2}})$ for all $j\neq 0$. Hence, we have $$k=\frac{|B_{-}(f)|}{p}.$$ The fact that $|B_{-}(f)|$ is an even integer implies $k$ is an even integer. On the other hand, we have $B_{-}(f)=-B_{-}(f)$ and $f^{*}(x)=f^{*}(-x)$. This implies that $D_{j}(f)=-D_{j}(f)$ and $|D_{j}(f)|$ is even for all $j \in \F_p$. By the similar arguments as in the previous case we deduce that $k$ is an odd integer. Hence, we arrive at a contradiction.

\end{case}
Since the case for $f$ is of type\;$(-)$ is similar, we skip it.
\end{proof}
\begin{definition}
Let $V$ be a subspace of  $\mathbb{F}_{p}^{n}$. $V$ is said to be a non-degenerate subspace with respect to the standart dot product in $\mathbb{F}_{p}^n$  if $\{x: x \in V| x \cdot y=0\;\text{for all}\; y \in V\}=\{0\}$.
\end{definition}
Let $W$ be a non-degenerate subspace of  $\mathbb{F}_{p}^{n}$. Then for all $x \in \mathbb{F}_{p}^{n}$ there exist unique  elements $u\in W$, $v\in W^{\perp}$ such that $x=u+v$. We denote this representation by the direct sum $\mathbb{F}_{p}^{n}=W \bigoplus W^{\perp} $.

Let $V$ be a non-degenerate vector space over $\F_p$ with respect to the standart dot product. Let $f:V\rightarrow\mathbb{F}_{p}$ be a $p$-ary function. Then the following equality is called Parseval's identity. $$\sum_{a \in V}|\hat{f}(a)|^2=|V|^2.$$
\begin{lemma}\label{sonn}
Let $f:\mathbb{F}_{p}^{n}\rightarrow\mathbb{F}_{p}$  be an non-weakly regular dual-bent function such that $f(x)=f(-x)$. If $B_{+}(f)$  is a non-degenerate subspace of dimension $r$, then there exist a partition $I_0^+(f),I_0^-(f)$ of $B_{+}(f)$  such that $$ B_{+}(f^*)=\cup_{u \in I_0^+(f)}(u+B_{+}(f)^{\perp}),$$ $$B_{-}(f^*)=\cup_{u \in I_0^-(f)}(u+B_{+}(f)^{\perp}).$$  If $B_{-}(f)$  is a non-degenerate subspace of  dimension $r$, then there exist a partition $I_1^+(f),I_1^-(f)$ of $B_{-}(f)$  such that $$ B_{+}(f^*)=\cup_{u \in I_1^+(f)}(u+B_{-}(f)^{\perp}),$$ $$B_{-}(f^*)=\cup_{u \in I_1^-(f)}(u+B_{-}(f)^{\perp}).$$  Moreover, for $p^n\equiv 1$  mod  $4$ if  $u \in I_0^+(f)$ then we have $$f\left|_{(u+B_{+}(f))^{\perp}}(x)=f(u),\right.$$ if $u \in I_1^-(f)$, then we have $$f\left|_{(u+B_{-}(f))^{\perp}}(x)=f(u).\right.$$ For $p^n\equiv 3$  mod  $4$, if  $u \in I_0^-(f)$ then we have $$f\left|_{(u+B_{+}(f))^{\perp}}(x)=f(u),\right.$$ if  $u \in I_1^+(f)$, then we have $$f\left|_{(u+B_{-}(f))^{\perp}}(x)=f(u).\right.$$
\end{lemma}

\begin{proof}

\begin{case}\label{wd++} $B_{+}(f)$  is a non-degenerate subspace of dimension $r$ and $p^n\equiv 1$  mod  $4$. 

Let $\alpha \in B_{+}(f^*)$.  There exist $u_0\in B_{+}(f)$, $v_0\in B_{+}(f)^{\perp}$ such that $\alpha=u_0+v_0$. By Proposition \ref{dual} we have $$ p^{\frac{n}{2}}\epsilon_{p}^{f(u_0+v_0)}=\sum_{x \in B_{+}(f)} \epsilon_{p}^{f^*(x)+(u_0+v_0) \cdot x}=\sum_{x \in B_{+}(f)} \epsilon_{p}^{f^*(x)+(u_0) \cdot x}.$$ Assume there exists $v_1 \neq v_0 \in B_{+}(f)^{\perp}$ such that $u_0+v_1 \in B_{-}(f^*).$ Then, by Proposition \ref{dual} we have $$\sum_{x \in B_{+}(f)} \epsilon_{p}^{f^*(x)+(u_0+v_1) \cdot x}=\sum_{x \in B_{+}(f)} \epsilon_{p}^{f^*(x)+(u_0) \cdot x}=0$$ which is a contradiction. Therefore, for all $v\in B_{+}(f)^{\perp}$, we have $u_0+v\in  B_{+}(f^*).$ This implies that $ p^{\frac{n}{2}}\epsilon_{p}^{f(u_0+v)}=\sum_{x \in B_{+}(f)} \epsilon_{p}^{f^*(x)+(u_0) \cdot x}$ for all $v\in B_{+}(f)^{\perp}$. In particular, $u_0\in  B_{+}(f^*).$ By Proposition \ref{dual}, we have $ p^{\frac{n}{2}}\epsilon_{p}^{f(u_0)}=\sum_{x \in B_{+}(f)} \epsilon_{p}^{f^*(x)+(u_0) \cdot x}$. Hence, for all $v\in B_{+}(f)^{\perp}$, we have $f(u_0+v)=f(u_0).$  By similar arguments, if $u_0+v_0\in  B_{-}(f^*)$, then for all $v\in B_{+}(f)^{\perp}$ we have $u_0+v\in  B_{-}(f^*).$ Let us define the subsets $I_0^+(f):= B_{+}(f)\cap B_{+}(f^*)$ and $I_0^-(f):= B_{+}(f)\cap B_{-}(f^*)$. Then, for all $u_0 \in I_0^+(f)$, the restriction of $f$ onto the coset $u_0+B_{+}(f)^{\perp}$ is the constant $f(u_0)$. Hence, the assertion of the lemma follows. 

\end{case}

\begin{case}\label{wd+-} $B_{+}(f)$  is a non-degenerate subspace of dimension $r$ and $p^n\equiv 3$  mod  $4$. 

Let $\alpha \in B_{-}(f^*).$  There exist $u_0\in B_{+}(f)$, $v_0\in B_{+}(f)^{\perp}$ such that $\alpha=u_0+v_0$. By Proposition \ref{dual} we have $$ -ip^{\frac{n}{2}}\epsilon_{p}^{f(u_0+v_0)}=\sum_{x \in B_{+}(f)} \epsilon_{p}^{f^*(x)+(u_0+v_0) \cdot x}=\sum_{x \in B_{+}(f)} \epsilon_{p}^{f^*(x)+(u_0) \cdot x}.$$ Assume there exists $v_1 \neq v_0 \in B_{+}(f)^{\perp}$ such that $u_0+v_1 \in B_{+}(f^*).$ Then, by Proposition \ref{dual} we have $$\sum_{x \in B_{+}(f)} \epsilon_{p}^{f^*(x)+(u_0+v_1) \cdot x}=\sum_{x \in B_{+}(f)} \epsilon_{p}^{f^*(x)+(u_0) \cdot x}=0$$ which is a contradiction. Therefore, for all $v\in B_{+}(f)^{\perp}$, we have $u_0+v\in  B_{-}(f^*).$ This implies that $ p^{\frac{n}{2}}\epsilon_{p}^{f(u_0+v)}=\sum_{x \in B_{+}(f)} \epsilon_{p}^{f^*(x)+(u_0) \cdot x}$ for all $v\in B_{+}(f)^{\perp}$. In particular, $u_0\in  B_{-}(f^*).$ By Proposition \ref{dual}, we have $ -ip^{\frac{n}{2}}\epsilon_{p}^{f(u_0)}=\sum_{x \in B_{+}(f)} \epsilon_{p}^{f^*(x)+(u_0) \cdot x}$. Hence, for all $v\in B_{+}(f)^{\perp}$, we have $f(u_0+v)=f(u_0).$ Therefore, for all $u_0 \in I_0^-(f)$ the restriction of $f$ onto the coset $u_0+B_{+}(f)^{\perp}$ is the constant $f(u_0)$. By similar arguments,  if $u_0+v_0\in  B_{+}(f^*)$, then for all $v\in B_{+}(f)^{\perp}$, we have $u_0+v\in  B_{+}(f^*).$ Hence, the assertion of the lemma follows. 

\end{case}

\begin{case}\label{wd-+} $B_{-}(f)$  is a non-degenerate subspace of dimension $r$ and $p^n\equiv 1$  mod  $4$. 

Let $\alpha \in B_{-}(f^*).$  There exist $u_0\in B_{-}(f)$, $v_0\in B_{-}(f)^{\perp}$ such that $\alpha=u_0+v_0$. By Proposition \ref{dual} we have $$ -p^{\frac{n}{2}}\epsilon_{p}^{f(u_0+v_0)}=\sum_{x \in B_{-}(f)} \epsilon_{p}^{f^*(x)+(u_0+v_0) \cdot x}=\sum_{x \in B_{-}(f)} \epsilon_{p}^{f^*(x)+(u_0) \cdot x}.$$ Assume there exists $v_1 \neq v_0 \in B_{-}(f)^{\perp}$ such that $u_0+v_1 \in B_{+}(f^*).$ Then, by Proposition \ref{dual} we have $$\sum_{x \in B_{-}(f)} \epsilon_{p}^{f^*(x)+(u_0+v_1) \cdot x}=\sum_{x \in B_{-}(f)} \epsilon_{p}^{f^*(x)+(u_0) \cdot x}=0$$ which is a contradiction. Therefore, for all $v\in B_{-}(f)^{\perp}$  we have $u_0+v\in  B_{-}(f^*).$ This implies that $ -p^{\frac{n}{2}}\epsilon_{p}^{f(u_0+v)}=\sum_{x \in B_{-}(f)} \epsilon_{p}^{f^*(x)+(u_0) \cdot x}$ for all $v\in B_{-}(f)^{\perp}$. In particular, $u_0\in  B_{-}(f^*).$ By Proposition \ref{dual}, we have $- p^{\frac{n}{2}}\epsilon_{p}^{f(u_0)}=\sum_{x \in B_{-}(f)} \epsilon_{p}^{f^*(x)+(u_0) \cdot x}$. Hence, for all $v\in B_{-}(f)^{\perp}$, we have $f(u_0+v)=f(u_0).$ By similar arguments, if $u_0+v_0\in  B_{+}(f^*)$, then for all $v\in B_{-}(f)^{\perp}$ we have $u_0+v\in  B_{+}(f^*).$ Let us define the subsets $I_1^+(f):= B_{-}(f)\cap B_{+}(f^*)$ and $I_1^-(f):= B_{-}(f)\cap B_{-}(f^*)$. Then, for all $u_0 \in I_1^-(f)$, the restriction of $f$ onto the coset $u_0+B_{-}(f)^{\perp}$ is the constant $f(u_0)$. Hence, the assertion of the lemma follows. 

\end{case}

\begin{case}\label{wd--} $B_{-}(f)$  is a non-degenerate subspace of dimension $r$ and $p^n\equiv 3$  mod  $4$. 

Let $\alpha \in B_{+}(f^*).$  There exist $u_0\in B_{-}(f)$, $v_0\in B_{-}(f)^{\perp}$ such that $\alpha=u_0+v_0$. By Proposition \ref{dual} we have $$ ip^{\frac{n}{2}}\epsilon_{p}^{f(u_0+v_0)}=\sum_{x \in B_{-}(f)} \epsilon_{p}^{f^*(x)+(u_0+v_0) \cdot x}=\sum_{x \in B_{-}(f)} \epsilon_{p}^{f^*(x)+(u_0) \cdot x}.$$ Assume there exists $v_1 \neq v_0 \in B_{-}(f)^{\perp}$ such that $u_0+v_1 \in B_{-}(f^*).$ Then, by Proposition \ref{dual} we have $$\sum_{x \in B_{-}(f)} \epsilon_{p}^{f^*(x)+(u_0+v_1) \cdot x}=\sum_{x \in B_{-}(f)} \epsilon_{p}^{f^*(x)+(u_0) \cdot x}=0$$ which is a contradiction. Therefore, for all $v\in B_{-}(f)^{\perp}$  we have $u_0+v\in  B_{+}(f^*).$ This implies that $i p^{\frac{n}{2}}\epsilon_{p}^{f(u_0+v)}=\sum_{x \in B_{-}(f)} \epsilon_{p}^{f^*(x)+(u_0) \cdot x}$ for all $v\in B_{-}(f)^{\perp}$. In particular, $u_0\in  B_{+}(f^*).$ By Proposition \ref{dual}, we have $ ip^{\frac{n}{2}}\epsilon_{p}^{f(u_0)}=\sum_{x \in B_{-}(f)} \epsilon_{p}^{f^*(x)+(u_0) \cdot x}$. Hence, for all $v\in B_{-}(f)^{\perp}$, we have $f(u_0+v)=f(u_0).$ Therefore, for all $u_0 \in I_1^+(f)$, the restriction of $f$ onto the coset $u_0+B_{-}(f)^{\perp}$ is the constant $f(u_0)$. By similar arguments, if $u_0+v_0\in  B_{-}(f^*)$, then for all $v\in B_{-}(f)^{\perp}$ we have $u_0+v\in  B_{-}(f^*).$ Then, the assertion of the lemma follows. 

\end{case}

\end{proof}

\begin{prp}\label{dualsize}
Let $f:\mathbb{F}_{p}^{n}\rightarrow\mathbb{F}_{p}$  be an non-weakly regular dual-bent function such that $f(x)=f(-x)$. 

\begin{itemize}

\item If $B_{+}(f)$  is a non-degenerate subspace of dimension $r$ over $\F_p$ and $p^n\equiv 1$  mod  $4$, or $B_{-}(f)$  is a non-degenerate subspace of dimension $r$ over $\F_p$ and $p^n\equiv 3$  mod  $4$. Then, $| B_{+}(f^*) |=p^r$,

\item If $B_{+}(f)$  is a non-degenerate subspace of dimension $r$ over $\F_p$ and $p^n\equiv 3$  mod  $4$, or $B_{-}(f)$  is a non-degenerate subspace of dimension $r$ over $\F_p$ and $p^n\equiv 1$  mod  $4$. Then, $| B_{-}(f^*) |=p^r$.

\end{itemize}
\end{prp}
\begin{proof}

\begin{itemize}

\item  $B_{+}(f)$  be a non-degenerate subspace of dimension $r$ over $\F_p$ and $p^n\equiv 1$  mod  $4$\\

Let $g$ be the restriction of $f^*$ into the subset $B_{+}(f)$. By Parseval's identity over $B_{+}(f)$  we have $$ \sum_{\alpha \in B_{+}(f)} |\hat{g}(\alpha)|^2=p^{2r},$$ where $\hat{g}(\alpha)= \sum_{x \in B_{+}(f)}\epsilon_{p}^{f^*(x)+\alpha \cdot x}$ for $\alpha \in B_{+}(f)$. Moreover, from Proposition \ref{dual} we have $$ \hat{g}(\alpha)=\left\{ \begin{array}{ll} p^{\frac{n}{2}}\epsilon_{p}^{f(\alpha)}  & \mbox{if $\alpha \in  I_0^+(f)$}; \\
	0 & \mbox{if $\alpha \in  I_0^-(f)$}.
	\end{array}
	\right.$$

From the arguments above we have $| I_0^+(f)|=p^{2r-n}$. Hence by Lemma \ref{sonn} we get $| B_{+}(f^*) |=p^r$.

\item $B_{-}(f)$  be a non-degenerate subspace of dimension $r$ over $\F_p$ and $p^n\equiv 3$  mod  $4$\\
Let $g$ be the restriction of $f^*$ into the subset $B_{-}(f)$. By Parseval's identity over $B_{-}(f)$  we have $$ \sum_{\alpha \in B_{-}(f)} |\hat{g}(\alpha)|^2=p^{2r},$$ where $\hat{g}(\alpha)= \sum_{x \in B_{-}(f)}\epsilon_{p}^{f^*(x)+\alpha \cdot x}$ for $\alpha \in B_{-}(f)$. Moreover, from Proposition \ref{dual} we have $$ \hat{g}(\alpha)=\left\{ \begin{array}{ll}i p^{\frac{n}{2}}\epsilon_{p}^{f(\alpha)}  & \mbox{if $\alpha \in  I_1^+(f)$}; \\
	0 & \mbox{if $\alpha \in  I_1^-(f)$}.
	\end{array}
	\right.$$
From the arguments above we have $| I_1^+(f)|=p^{2r-n}$. Hence, by Lemma \ref{sonn} we get $| B_{+}(f^*) |=p^r$.

\item If $B_{+}(f)$  be a non-degenerate subspace of dimension $r$ over $\F_p$ and $p^n\equiv 3$  mod  $4$\\
Let $g$ be the restriction of $f^*$ into the subset $B_{+}(f)$. By Parseval's identity over $B_{+}(f)$  we have $$ \sum_{\alpha \in B_{+}(f)} |\hat{g}(\alpha)|^2=p^{2r},$$ where $\hat{g}(\alpha)= \sum_{x \in B_{+}(f)}\epsilon_{p}^{f^*(x)+\alpha \cdot x}$ for $\alpha \in B_{+}(f)$. Moreover, from Proposition \ref{dual} we have $$ \hat{g}(\alpha)=\left\{ \begin{array}{ll} -ip^{\frac{n}{2}}\epsilon_{p}^{f(\alpha)}  & \mbox{if $\alpha \in  I_0^-(f)$}; \\
	0 & \mbox{if $\alpha \in  I_0^+(f)$}.
	\end{array}
	\right.$$

From the arguments above we have $| I_0^-(f)|=p^{2r-n}$. Hence by Lemma \ref{sonn} we get $| B_{-}(f^*) |=p^r$.

\item $B_{-}(f)$  be a non-degenerate subspace of dimension $r$ over $\F_p$ and $p^n\equiv 1$  mod  $4$

Let $g$ be the restriction of $f^*$ into the subset $B_{-}(f)$. By Parseval's identity over $B_{-}(f)$  we have $$ \sum_{\alpha \in B_{-}(f)} |\hat{g}(\alpha)|^2=p^{2r},$$ where $\hat{g}(\alpha)= \sum_{x \in B_{-}(f)}\epsilon_{p}^{f^*(x)+\alpha \cdot x}$ for $\alpha \in B_{-}(f)$. Moreover, from Proposition \ref{dual} we have $$ \hat{g}(\alpha)=\left\{ \begin{array}{ll}-p^{\frac{n}{2}}\epsilon_{p}^{f(\alpha)}  & \mbox{if $\alpha \in  I_1^-(f)$}; \\
	0 & \mbox{if $\alpha \in  I_1^+(f)$}.
	\end{array}
	\right.$$
From the arguments above we have $| I_1^-(f)|=p^{2r-n}$. Hence, by Lemma \ref{sonn} we get $| B_{-}(f^*) |=p^r$.

\end{itemize}

\end{proof}

A function $f:\mathbb{F}_{p}^{n}\rightarrow\mathbb{F}_{p}$ is called $s$-\textit{plateaued} if $|{\hat{f}(\alpha)}|=p^{\frac{n+s}{2}}\;\;\text{or}\;\;0$ for all $\alpha \in \mathbb{F}_{p}^n.$ The Walsh spectrum of $s$-\textit{plateaued} functions is given as follows (see \cite{jjy}),
$$\hat{f}(\alpha)= \left\{ \begin{array}{ll}
\pm p^{\frac{n+s}{2}}\epsilon_{p}^{f^*(\alpha)},0 & \mbox{if $n+s$ even or $n+s$ odd and $p\equiv 1$ mod $ 4$}; \\
\pm i p^{\frac{n+s}{2}}\epsilon_{p}^{f^*(\alpha)},0 & \mbox{if $n+s$ odd and $p\equiv 3$  mod  $4$}.
\end{array}
\right. $$

We denote the support of  $\hat{f}$  by $\mathrm{Supp}(\hat{f})$ and it is defined by $\mathrm{Supp}(\hat{f}):=\{\alpha:\alpha\in \mathbb{F}_{p}^{n}\mid \hat{f}(\alpha)\ne 0 \}.$  For $v\in \mathbb{F}_{p}^{n}$  let  $D_vf$ be the derivative function $D_vf(x):\mathbb{F}_{p}^{n}\rightarrow\mathbb{F}_{p}$ given by $D_vf(x)=f(x+v)-f(x)$. 

\begin{definition} 
A function $f:\mathbb{F}_{p}^{n}\rightarrow\mathbb{F}_{p}$ is called partially bent if the following property holds: For $v\in \mathbb{F}_{p^n}$, the derivative function is either a balanced or a constant. 
\end{definition}
Note that \textit{partially bent} functions are special subclass of \textit{plateaued} functions, and most of the known \textit{plateaued} functions are \textit{partially bent}. In the following subsection, we use the arguments that are given in \cite{op1}.

\subsection{GMMF Bent Functions}
Let $p$ be an odd prime and $F:\mathbb{F}_{p}^{m}\times\mathbb{F}_{p}^{s}\rightarrow\mathbb{F}_{p}$ be the map $(x,y)\rightarrow f_y(x),$  where $f_y:\mathbb{F}_{p}^{m}\rightarrow \mathbb{F}_{p}$ is an $s$-\textit{plateaued} function for each $y\in \mathbb{F}_{p}^{s}$ such that $\mathrm{Supp}(\hat{f_i})\cap \mathrm{Supp}(\hat{f_j})=\emptyset$ for $i\neq j$, $i,j\in \mathbb{F}_p^s.$ In \cite{Cesmelioglu2013b}, the authors  showed that $F$ is a bent function.

The \textit{bent} functions of the form $F(x,y)=f_y(x)$ are called \textit{GMMF} (Generalized Maiorana-McFarland) (see \cite{Cesmelioglu2013a}). From now on, all the plateaued functions we consider are partially bent. It is known that all partially bent functions can be written as sum of a bent function and an affine function. Let $f^{(a)}:\mathbb{F}_{p}^{m}\rightarrow\mathbb{F}_{p}$ be bent for all $a\in \F_{p^s}$ and $f_a:\mathbb{F}_{p}^{m}\times\mathbb{F}_{p}^{s}\rightarrow\mathbb{F}_{p}$ be the map $(x,y)\rightarrow f^{(a)}(x)+a \cdot y$. Then the function $F:\mathbb{F}_{p}^{m}\times\mathbb{F}_{p}^{s}\times\mathbb{F}_{p}^{s}\rightarrow\mathbb{F}_{p}$ defined by \be F(x,y,z)=f_z(x,y)=f^{z}(x)+z \cdot y \label{gmmf}\ee  belongs to GMMF class.
The Walsh transform of $F$ at $(\alpha,\beta, \gamma)$ is given by

$$\begin{array}{lll}
\widehat{F}(\alpha,\beta,\gamma)&=&\sum_{x\in\mathbb{F}_p^m}\sum_{y\in\mathbb{F}_p^s}\sum_{z\in\mathbb{F}_p^s}\epsilon_p^{F(x,y,z)-\alpha \cdot x-\beta \cdot y-\gamma \cdot z}
\\
&=&\sum_{x\in\mathbb{F}_p^m}\epsilon_{p}^{f^{(z)}(x)-\alpha \cdot x}\sum_{y\in\mathbb{F}_p^s}\epsilon_{p}^{y \cdot (z-\beta)}\sum_{z\in\mathbb{F}_p^s}\epsilon_p^{-\gamma \cdot z}
\\
&=& p^s\epsilon_{p}^{-\gamma \cdot\beta}\widehat{f^{(\beta)}}(\alpha).
\end{array}$$  Then we have,
\be \label{typ}\widehat{F}(\alpha,\beta,\gamma)=\xi_{\alpha,\beta}p^{\frac{m+2s}{2}}\epsilon_{p}^{{f^{(\beta)}}^*(\alpha)-\gamma \cdot \beta}\ee which follows from $\widehat{f^{(\beta)}}(\alpha)=\xi_{\alpha,\beta}p^{\frac{m}{2}}\epsilon_{p}^{{f^{(\beta)}}^*(\alpha)},$ where $\xi_{\alpha,\beta}\in \{\pm 1,\pm i\}.$ Hence we have \be F^*(x,y,z)={f^{(y)}}^*(x)-y \cdot z. \label{gmmfdual}\ee
\begin{obsrv} \label{nwrb}
 $F$ is weakly regular if $f^{(z)}$ is weakly regular bent with the same sign for all $z\in \mathbb{F}_p^s$  in their Walsh coefficients. $F$ is non-weakly regular bent if $f^{(z)}$ is weakly regular bent for all $z\in \mathbb{F}_p^s$ and there are $ z_1,z_2\in \mathbb{F}_p^s$ such that $f^{(z_1)}$ and $f^{(z_2)}$  have different signs in their Walsh coefficients or there exists $z\in \mathbb{F}_p^s$ such that $f^{(z)}$ is non-weakly regular bent.
\end{obsrv}

\par Let $F\in$ \textit{GMMF} be a non-weakly regular bent function. We determine the structure of the sets $B_{+}(F)$ and $B_{-}(F)$ in the case of $f^{(z)}$ is weakly regular bent for all $z\in \mathbb{F}_p^s$. Since sign of the Walsh coefficients of a weakly regular bent function doesn't change with respect to $\alpha$, in this case Equation \eqref{typ} reduces to \be \label{typp}\widehat{F}(\alpha,\beta,\gamma)=\xi_{\beta}p^{\frac{n+s}{2}}\epsilon_{p}^{{f^{(\beta)}}^*(\alpha)-\gamma \cdot \beta}.\ee

By the observation above, one can partition $\mathbb{F}_{p}^s$  into two subsets as $W^{+}(F):=\{z: z \in \mathbb{F}_{p}^{s} | f^{(z)}\; \text{is of type}
\;(+)\}$ and $W^{-}(F):=\{z: z\in \mathbb{F}_{p}^{s} | f^{(z)}\; \text{is of type}\;(-)\},$ where $F:\mathbb{F}_{p}^m\times \mathbb{F}_{p}^s\times \mathbb{F}_{p}^s\rightarrow \mathbb{F}_p$ is given by $F(x,y,z)= f_z(x,y).$ Then by Equation (\ref{typp}) we deduce that \be \label{cos} B_{\pm}(F)=\mathbb{F}_{p}^{m}\times W^{\pm}(F) \times \mathbb{F}_{p}^s.\ee
\begin{rem} \label{rem1} Note that Equation (\ref{gmmfdual}) implies $F^*$ is also bent and belongs to the \textit{GMMF} class if and only if  ${f^{(y)}}^*(x)$ is bent for all $y \in \mathbb{F}_{p}^{s} $. Observe that if  $f^{(z)}$ is weakly regular bent for all $z\in \mathbb{F}_p^s$ then  $F^*$ is bent. \end{rem} 
Let $F^*$ be bent. The Walsh transform of $F^*$ at $(\alpha,\beta, \gamma)$ is given by

$$ \begin{array}{lll}
\widehat{F^*}(\alpha,\beta,\gamma)&=&\sum_{x\in\mathbb{F}_p^m}\sum_{y\in\mathbb{F}_p^s}\sum_{z\in\mathbb{F}_p^s}\epsilon_p^{F^*(x,y,z)-\alpha \cdot x-\beta \cdot y-\gamma \cdot z}
\\
&=&\sum_{x\in\mathbb{F}_p^m}\epsilon_{p}^{{f^{(y)}}^*(x)-\alpha \cdot x}\sum_{y\in\mathbb{F}_p^s}\epsilon_{p}^{-y \cdot \beta}\sum_{z\in\mathbb{F}_p^s}\epsilon_p^{-z \cdot (\gamma+y)}
\\
&=& p^s\epsilon_{p}^{\gamma \cdot\beta}\widehat{{f^{(-\gamma)}}^*}(\alpha). 
\end{array} $$
In the case of $f^{(z)}$ is weakly regular bent for all $z\in \mathbb{F}_p^s$, from Equation (\ref{wrdual}) we obtain,
\be \label{typpp}\widehat{F^*}(\alpha,\beta,\gamma)=\xi_{\beta}^{-1}p^{\frac{m+2s}{2}}\epsilon_{p}^{{f^{(-\gamma)}}(-\alpha)+\gamma \cdot \beta},\ee where $\widehat{{f^{(-\gamma)}}^*}(\alpha)=\xi_{\beta}^{-1}p^{\frac{m}{2}}\epsilon_{p}^{{f^{(-\gamma)}}(-\alpha)}$ and $\xi_{\beta}\in \{\pm 1,\pm i\}.$ 

Hence we have $$ \begin{array}{lll} F^{**}(x,y,z)&=&{f^{(-z)}}^{**}(x)+y \cdot z
\\
&=& f^{(-z)}(-x)+y \cdot z.
\end{array} $$ 

Then, from Equation (\ref{typpp}) we have \begin{eqnarray} \label{coss} B_{+}(F^*)&=&\mathbb{F}_{p}^{m} \times \mathbb{F}_{p}^s\times(-W^{\pm}(F)), \\ \label{cosss} B_{-}(F^*)&=&\mathbb{F}_{p}^{m} \times \mathbb{F}_{p}^s\times (-W^{\pm}(F)),\end{eqnarray} where $(\pm)$ sign depends on $p$ and $m$ (see Remark \ref{rem0}).

\par Character theory of finite groups are a very useful tool to study exponential sums over finite fields, the reader may refer to anygood textbook on character theory for basic facts and notations, see  for instance \cite{isac}. The combination of character theory of finite groups and elementary number theory including cyclotomic fields are also very useful to evaluate the Walsh spectrum of $p$-ary functions. For further reading on elementary number theory and cyclotomic fields, we refer to \cite{hardy} and  \cite{wash} respectively.

\subsection{Basic Character Theory.}

The functions $\chi_{j} : \F_{p}^n \rightarrow \C^{\star},\;\; j\in \F_{p}^n$, defined by
$$ \chi_{j}(x)=\epsilon_{p}^{j \cdot x}$$ are all additive characters of $\F_{p}^n$. 
Let $\widehat{G}$ denotes the character group of an finite abelian group $G$ and $\chi_0$ be the trivial character. We identify a subset $A$ of $G$ with the group ring element $\sum_{x\in A}x$, which will also be denoted by $A$. By linearity, we extend each character $\chi \in G$ to a homomorphism from $\C[G]$ to $\C$, where $\C[G]$ is the group ring of $G$ over the field of complex numbers. We still denote this homomorphism by $\chi$, i.e., $\chi(A)=\sum_{x\in A}\chi(x)$.

\subsection{Elementary Number Theory.}

Let $a$ be a positive integer and $p$ be an odd prime number.  Let $ a\equiv\tilde{a}\;(mod\;p)$.
The \textit{Legendre symbol} is defined as
\begin{displaymath}
\left( \frac{a}{p} \right) = \left\{ \begin{array}{ll}
\,\,\,\, 0 & \textrm{if }  \tilde{a}=0;\\
\,\,\, \, 1 & \textrm{if } \sqrt{\tilde{a}} \in \F_p^\star;\\
-1 & \textrm{if }  \sqrt{\tilde{a}} \notin \F_{p}^{\star}.
\end{array} \right.
\end{displaymath}
Let $p$ be an odd prime number. The quadratic Gauss sum is defined as
\be \sum_{i\in\mathbb{F}_p^{\star}}\left( \frac{i}{p} \right)\epsilon_p^i=\left\{ \begin{array}{ll} \sqrt{p}  & \mbox{if  $ p \equiv 1\;(mod\;4)$}; \\
	i\sqrt{p} & \mbox{if  $ p \equiv 3\;(mod\;4)$}.
\end{array}
\right. \label{gauss}\ee

In the following parts, we use the notation as in \cite{mos} to define the cyclotomic fields and linear codes.

\noindent \textbf{Cyclotomic field.}
Let $p$ be an odd prime. The extension field $\Q(\epsilon_p)$ of the field of rational numbers is called $p$-th cyclotomic field. 
The ring of integers in $\mathbb{Q}(\epsilon_p)$ is defined as $\mathcal {O}_{\mathbb{Q}(\epsilon_p)}:=\mathbb{Z}(\epsilon_p)$. Let $I \subset \F_p$ be a subset of size $p-1$ then, $\{\epsilon_p^i : i \in I\}$ is an integral basis  of $\mathcal {O}_{\mathbb{Q}(\epsilon_p)}$.
The field extension $\mathbb{Q}(\epsilon_p)/\mathbb{Q}$ is Galois of degree $p-1$, and  the Galois group 
$$Gal (\mathbb{Q}(\epsilon_p)/\mathbb{Q})=\{\sigma_a : a \in \F_p^{\star}\},$$
where the automorphism $\sigma_a$ of $\mathbb{Q}(\epsilon_p)$ is defined by $\sigma_a(\epsilon_p)=\epsilon_p^a$.

\subsection{Linear codes.}
Let $p$ be a prime number and $n,\;k \in \Z^+$. A linear code $\mathcal {C}$  of length $n$ and dimension $k$ over $\mathbb {F}_{p}$ is a $k$-dimensional linear subspace of $\F_p^n$, denoted by $[n,k]_{p}$. The elements of $\sC$  are called  \textit{codewords}.
A linear code $\mathcal{C}$  of length $n$ and dimension $k$ over $\mathbb {F}_{p}$ with minimum Hamming distance $d$ is denoted by  $[n,k,d]_{p}$. The Hamming weight of  a vector $ u=(u_0,\ldots, u_{n-1})\in \mathbb {F}_{p}^n$, denoted by $wt( u)$, is the size of its support  defined as $$\supp(u)=\{0 \leq i\leq n-1: u_i\not=0 \}.$$ It is clear that the minimum Hamming weight of nonzero codewords of a linear code $\sC$ correponds to minimum Hamming distance of $\sC$. 

Let $E_a$ be denote the number of codewords with Hamming weight $a$ in $\sC$ of length $n$. Then, $(1,E_1, \ldots, E_n)$ is the \textit{ weight distribution} of $\sC$ and the polynomial $1+E_1y + \cdots + E_ny^n$ is called the \textit{ weight enumerator} of $\sC$. 
The code $\sC$ is called a \textit{$t$-weight code} if  $|\{a:a \in \F_p^{\star}| E_a \ne 0 \}|=t$.

\par In the following two sections, we construct three-weight linear codes based on the second generic construction. Although the general idea of the construction method employed is a classical one, we use non-weakly regular bent functions to construct linear codes over finite fields.

\section{Three-weight ternary linear codes on $\mathbf{B_{+}(f)}$}\label{SectionConstruction}

Let $S$ be a non empty subset of $\mathbb{F}_{p}^{n}.$ The following generic construction method was employed
in \cite{ding1} for obtaining linear codes with a few weights.
\begin{equation}\label{defCode}
\mathcal {C}_{S}=\{c_{u}=\left(u \cdot x_1, u \cdot x_2,\ldots, u \cdot x_{k}\right) \, : \; u\in\mathbb {F}_{p^n}\},
\end{equation}
where $x_1, \ldots, x_{k}$ are the non-zero elements of $S$  and  $c_{u}$ denotes a codeword of $\mathcal {C}_{S}$. $S$ is called the \textit{defining set} of the code $\mathcal {C}_{S}.$ Clearly, the length of the linear code $\mathcal {C}_{S}$ is $k$ and its dimension is at most $n$.  For any $u \in \mathbb{F}_{p}^{n},$ the weight of the codeword $c_{u}$ of $\mathcal {C}_{S}$ can be written as follows: \be  \begin{array}{lll} wt(c_{u})&=& k-\frac{1}{p}\sum_{x \in S}\sum_{y\in \mathbb {F}_{p}}\epsilon_{p}^{y(u \cdot x)} \\
	&=&\frac{p-1}{p}k-\frac{1}{p}\sum_{y\in \mathbb {F}_{p}^{*}}\sigma_y(\chi_u(S)).
\end{array} \label{wt}\ee
Let $D$ be  a subset of $\F_p^n$. Any function  $f :  D \longrightarrow  \mathbb {F}_{p}$ is said to be \textit{balanced} over $\F_p$ if 
$f$ takes every value of $\F_p$ the same number of times. 

\begin{fact}\label{bal}
	If $D$ is a subspace of $\F_p^n$ and $j \notin D^{\perp}$, then it is well-known that $j \cdot x$ is balanced over $\F_p$.
\end{fact}

From now on we further assume that $f:\mathbb{F}_{3}^{n}\rightarrow\mathbb{F}_{3}$  be a non-weakly regular dual-bent function such that $f(x)=f(-x)$. Note that $f^{*}$ is non-weakly regular bent and $f^{**}(x)=f(x)$ (see\; \cite[Corollary 3.1 and Theorem 3.1]{op}). Since $B_{+}(f)$ is a vector space and $f^*(x)=f^*(-x)$ from Proposition \ref{f*0}, then we have $C_i(f)=-C_i(f)$ for all $i \in \F_3$ which implies $\chi_u(C_i(f)) \in \Z$ for all $i\in \F_3$, $u\in \F_3^n$. Then, by Equation (\ref{wt}), for all $i\in \F_3$ we have \be \label{wts} wt(c_{u})=\frac{2}{3}\left(|C_i(f)^{\star}|-\chi_u(C_i(f)^{\star})\right).\ee
\\

\par In Section \ref{SectionConstruction}, we assume $B_{+}(f)$ is an $\mathbb {F}_{3}$-vector space with $\Dim(B_{+}(f))\ge \lfloor \frac{n}{2}\rfloor+1.$ Put $\dim(B_{+}(f))=r$. Since $0 \in B_{+}(f),$ then $f$ is of type $(+)$.
Let $n$ be even. Then, by Proposition \ref{dualtype} $f^*$ is also of type $(+)$. We next determine the $|C_i(f)|$ for each $i\in \F_3$. Put $f(0)=j_0$. Then, from Proposition \ref{dual} we have $3^{\frac{n}{2}}\epsilon_{3}^{j_0}=|C_0(f)|+|C_1(f)|\epsilon_{3}+|C_2(f)|\epsilon_{3}^2$. Since, $\{\epsilon_{3}^{j_0+1},\epsilon_{3}^{j_0+2}\}$ be an integral basis for $\Q(\epsilon_{3})$, then we have $|C_{j_0}(f)|-3^{\frac{n}{2}}=|C_{j_0+1}(f)|=|C_{j_0+2}(f)|,$ where the summations $j_0+1$ and $j_0+2$ are modulo $3$. On the other hand, we have $|C_0(f)|+|C_1(f)|+|C_2(f)|=3^r$. Hence, we get \be \label{car+}  |C_{j_0}(f)|=3^{r-1}-3^{\frac{n}{2}-1}+3^{\frac{n}{2}},\;\; |C_{j_0+1}(f)|=|C_{j_0+2}(f)|=3^{r-1}-3^{\frac{n}{2}-1}.  \ee

\begin{prp}\cite[Proposition 1]{Cesmelioglu2013a} \label{w+-}

 Let $f(x_1,x_2,\dots, x_n)=d_1x_1^2+d_2x_2^2+\dots+d_nx_n^2$ be a quadratic bent function from $\mathbb{F}_{p}^n$ to $\mathbb{F}_p$. Let $\Delta:=\prod_{i=1}^{n}d_i$, and let $\eta$ denote the quadratic character of $\mathbb{F}_{p}$. The Walsh spectrum of $f$ is given by
$$\hat{f}(\alpha)= \left\{ \begin{array}{ll}
\eta(\Delta) p^{n/2}\epsilon_{p}^{f^*(\alpha)} & \mbox{if $p\equiv 1$ mod $ 4$}; \\ \eta(\Delta )i^n p^{n/2} \epsilon_{p}^{f^*(\alpha)} & \mbox{if $p\equiv 3$ mod $4$}.
\end{array}
\right.$$ 
\end{prp}

\begin{cor} \label{qtype}
For any odd prime $p$ and $n\in Z^+$ there exist quadratic bent functions of type\;$(+)$ and type\;$(-)$.
\end{cor}

\begin{proof}
The desired conclusion trivially follows from Proposition \ref{w+-}.
\end{proof}

\begin{rem}\label{rem2}
By Observation \ref{nwrb} and Corollary \ref{qtype} one can easily obtain infinitely many weakly regular bent functions of any type and varying degrees as follows. Let  $F:\mathbb{F}_{p}^{r}\times\mathbb{F}_{p}^{s}\times\mathbb{F}_{p}^{s}\rightarrow\mathbb{F}_{p}$ be defined by Equation \eqref{gmmf} such that $f^{(z)}$ is quadratic bent function of  type \;$(+)$ (resp. type  \;$(-)$ ) for all $z\in \mathbb{F}_p^s$. Then $F$ is weakly regular bent of  type \;$(+)$ (resp. type  \;$(-)$ ) by Equation \eqref{typ}. Moreover, if  $\sum_{z\in \mathbb{F}_p^s}f^{(z)}$ is quadratic (affine) then by \cite[Proposition 1]{Cesmelioglu2013c} we have $deg(F)=(p-1)s+2 (deg(F)=(p-1)s+1)$. Observe that the degree of $F$ varies as $s$ changes. Since the dimension $m$ equals $r+2s$,  for any $m\ge3$ and $p$ we can obtain non-quadratic weakly regular bent functions of arbitrary type. Observe that if $f^{(z)}$ is homogenous quadratic bent function for all $z\in \mathbb{F}_p^s$ with  $f^{(z)}=f^{(-z)}$, then by Equation \eqref{gmmf} we have $F(x,y,z)=F(-x,-y,-z)$.  Thanks to the \textit{GMMF} class (\cite{Cesmelioglu2013a}), for any $j_0 \in \F_3$  by recursively applying the procedure we describe above we can obtain  infinitely many non-weakly regular dual-bent functions such that $F(0)=j_0$, $F(x)=F(-x),$ and $B_{\pm}(F)$ is a non-degenerate vector subspace. One can construct such a function as follows. Let $F:\mathbb{F}_{3}^{m}\times\mathbb{F}_{3}^{s}\times\mathbb{F}_{3}^{s}\rightarrow\mathbb{F}_{3}$ be a bent function defined by Equation \eqref{gmmf} such that $f^{(z)}$ is weakly regular (not necessarily quadratic) for all $z \in \F_{3^s}$ and there exist $z_1,z_2 \in \F_{3^s}$  with $f^{(z_1)}$ and $f^{(z_2)}$ are of different types. Then, $F$ is non-weakly regular bent and $F^*$ is also bent (see Remark \ref{rem1}). Then Equation \eqref{gmmfdual} implies that $B_{\pm}(f)$ is a vector space if and only if $W^{\pm}(F)$ is a vector subspace of $\mathbb{F}_{3}^{s}$ respectively. By recursively applying the procedure we describe above, for any $s\ge 1$ and $z \in \mathbb{F}_{3}^{s}$ we can find a weakly regular bent function $f^{(z)}:\mathbb{F}_{3}^{m}\rightarrow\mathbb{F}_{3}$ of arbitrary type and varying degrees.  For any $0\le s_1 < s$ let us choose a $s_1$-dimensional vector subspace of $U$ of $\mathbb{F}_{3}^{s}$. If, for $z \in U$ we set  $f^{(z)}$ is of  type \;$(+)$ (resp. type  \;$(-)$ ) and  for $z \notin U$ we set  $f^{(z)}$ is of  type \;$(-)$ (resp. type  \;$(+)$ ), then  $W^{+}(F)$ (resp. $W^{-}(F)$) becomes a vector subspace of dimension $s_1$. Moreover, by Equation \eqref{gmmf}, we have $F(0)=f^{(0)}(0)$ and $F(-x,-y,-z)=f^{(-z)}(-x)+y\cdot z$. If we set $f^{(z)}=f^{(-z)}$ for all $z\in \mathbb{F}_{3}^{s}$ and $f^{(z)}(x)=f^{(z)}(-x)$ for all $x\in \mathbb{F}_{3}^{m},\;z\in \mathbb{F}_{3}^{s}$ (for any $m\ge 1$ and $p$, such functions with varying degrees exists by the arguments given above), then we have $F(x,y,z)=F(-x,-y,-z)$. Furthermore, for any $j_0 \in \F_3$ one can easily find a weakly regular bent function $g:\mathbb{F}_{3}^{m}\rightarrow\mathbb{F}_{3}$ such that $g(0)=j_0$. Hence, if we set $f^{(0)}=g,$ then we have $F(0)=j_0$. It is easy to see that if  $W^{\pm}(F)$ is a non-degenerate subspace then $B_{\pm}(f)$ is non-degenerate as well. By the arguments above we can set $W^{\pm}(F)$ as an arbitrary proper subspace of $\mathbb{F}_{3}^{s}$. Hence, we can eleminate the degeneracy case easily.
\end{rem}

\par For $f(0)=j_0$, let us take $S=C_{j_0}(f)$. Observe that the linear code $\sC_{C_{j_0}(f)}$ of length $|C_{j_0}(f)|-1$ over $\F_3$ defined by $(\ref{defCode})$ is at most $r$-dimensional.

\begin{prp} \label{kernel}
	The linear code $\sC_{C_{j_0}(f)}$ of length $3^{r-1}-3^{\frac{n}{2}-1}+3^{\frac{n}{2}}-1$ over $\F_3$ defined  by $(\ref{defCode})$ is $r$-dimensional.
\end{prp}

\begin{proof}
	Let $\theta: \mathbb{F}_{3}^{n}\rightarrow\mathbb{F}_{3}^{3^{r-1}-3^{\frac{n}{2}-1}+3^{\frac{n}{2}}-1}$ be the map defined by $u \rightarrow c_u.$ Since, by definition we have $C_{j_0}(f) \subset B_{+}(f)$, then it is clear that $\big{(}B_{+}(f)\big{)}^{\perp} \subseteq \Kernel(\theta).$ Now, we will show that  $\Kernel(\theta)=\big{(}B_{+}(f)\big{)}^{\perp}$ by proving $\Kernel(\theta)\subseteq \big{(}B_{+}(f)\big{)}^{\perp}$. Let $u \in  \big{(}\big{(}B_{+}(f)\big{)}^{\perp}\big{)}^C \bigcap \Kernel(\theta) $. Clearly $u \ne 0$. If $u \notin \big{(}B_{+}(f)\big{)}^{\perp},$ then from Fact \ref{bal} we have $|\{x: x \in B_{+}(f)|u \cdot x=0 \}|=3^{r-1}$. On the other hand,  $ u \in \Kernel(\theta)$ implies $u\cdot x=0$ for all $x \in C_{j_0}(f)$. For $n\ge 2$ it is clear that $|C_{j_0}(f)|=3^{r-1}-3^{\frac{n}{2}-1}+3^{\frac{n}{2}}-1>3^{r-1}$.  But, $C_{j_0}(f) \subset B_{+}(f)$ implies $|\{x: x \in B_{+}(f)|u \cdot x=0 \}|>3^{r-1}$ which gives a contradiction. Hence, we prove that $\big{(}\big{(}B_{+}(f)\big{)}^{\perp}\big{)}^C \bigcap \Kernel(\theta)=\emptyset$ which shows that $\Kernel(\theta)\subseteq\big{(}B_{+}(f)\big{)}^{\perp}$. Since, $|B_{+}(f)|=3^{r}$, then $|\big{(}B_{+}(f)\big{)}^{\perp}|=3^{n-r}.$ Therefore, by the isomorphism  $\bar{\theta}: \mathbb{F}_{3}^{n}/\Kernel(\theta) \rightarrow \Image(\theta)$, we have $| \Image(\theta)|=3^{r}$.
\end{proof}

\begin{rem}\label{rem3}
Let $f:\mathbb{F}_{3}^{n}\rightarrow\mathbb{F}_{3}$  be a non-weakly regular dual-bent function with $f(0)=j_0$ and $B_{+}(f)$ be a vector space of dimension $r$ over $\F_3$. For some $i\in F_3$ let us consider the linear code $\sC_{C_{i}(f)}$ of length over $\F_3$ defined by $(\ref{defCode})$. As it can seen from the proof of Proposition \ref{kernel}, to guarantee that $\Dim(\sC_{C_{i}(f)})=r$, we need to have $|C_{i}(f)^{\star}|>3^{r-1}$, and from Equation \eqref{car+} it is possible if and only if $i=j_0$.
	\end{rem}

\begin{teo} \label{even}
	
	Let $n$ be an even integer and $f:\mathbb{F}_{3}^{n}\rightarrow\mathbb{F}_{3}$  be a non-weakly regular dual-bent function such that $f(0)=j_0$, and $f(x)=f(-x)$. Let $B_{+}(f)$  be an $r$-dimensional non-degenerate $\mathbb {F}_{3}$-vector space with $r\ge \frac{n}{2}+1.$  Then, the code $\mathcal {C}_{C_{j_0}(f)}$ whose codewords $c_{u}$ are defined by Equation (\ref{defCode}) is a three weight ternary linear code with parameters $[3^{r-1}-3^{\frac{n}{2}-1}+3^{\frac{n}{2}}-1,r,23^{r-2}]_{3}$. The codeword $c_{u}$ has zero-weight if $u \in  \big{(}B_{+}(f)\big{)}^{\perp}$ i.e., $u \in \Kernel(\theta)$. For $u \notin  \big{(}B_{+}(f)\big{)}^{\perp}$ the non-zero weight codewords are as follows.
	
	$$ wt (c_{u})=\left\{ \begin{array}{lll}
	23^{r-2}& \mbox{if $u \in B_{+}(f^*)$ and $f(u)=j_0$} ;\\
	2(3^{r-2}+3^{\frac{n}{2}-1}) & \mbox{if $u \in B_{+}(f^*)$ and $f(u) \ne j_0$};\\
	2(3^{r-2}-3^{\frac{n}{2}-2}+3^{\frac{n}{2}-1}) & \mbox{if $u \in B_{-}(f^*)$ }.
	\end{array}
	\right. $$
	Moreover, the weight distribution of $\mathcal {C}_{C_{j_0}(f)}$ is as in Table \ref{tabloo1}.
\end{teo}

\begin{proof}

Let us now evaluate the non-zero weights of codewords in $\sC_{C_{j_0}(f)}$. For $c_{u}\in \mathcal {C}_{C_{j_0}(f)}$ we have the following.

	\begin{itemize}
		
		\item $u \in B_{+}(f)^{\perp}$
		\\
		By Equation \eqref{preim}, we have $ C_{j_0}(f) \subset B_{+}(f)$ which implies $wt(c_{u})=0.$
		\item $u \notin B_{+}(f)^{\perp} $ and $u \in B_{+}(f^*).$
		\\
		From Proposition \ref{dual} we have $3^{\frac{n}{2}}\epsilon_{3}^{f(u)}=\chi_u(C_0(f))+\chi_u(C_1(f))\epsilon_{3}+\chi_u(C_2(f))\epsilon_{3}^2$. Let $f(u)=j.$ Using the fact that $\{\epsilon_{3}^{j+1},\epsilon_{3}^{j+2}\}$ being an integral basis for $\Q(\epsilon_{3})$ we have $\chi_u(C_j(f))-3^{\frac{n}{2}}=\chi_u(C_{j+1}(f))=\chi_u(C_{j+2}(f)),$ where the summations $j+1$ and $j+2$ are modulo $3$. By Fact \ref{bal}, we also have $\chi_u(C_0(f))+\chi_u(C_1(f))+\chi_u(C_2(f))=0.$ Hence, we get $\chi_u(C_j(f))=3^{\frac{n}{2}}-3^{\frac{n}{2}-1}$ and $\chi_u(C_{j+1}(f))=\chi_u(C_{j+2}(f))=-3^{\frac{n}{2}-1}$. From Equation (\ref{wts}) we have $$wt(c_{u}) \in \{23^{r-2},2(3^{r-2}+3^{\frac{n}{2}-1})\}.$$
		
		\item  $u \notin B_{+}(f)^{\perp} $ and $u \in B_{-}(f^*).$
		\\
		From Proposition \ref{dual} we have $0=\chi_u(C_0(f))+\chi_u(C_1(f))\epsilon_{3}+\chi_u(C_2(f))\epsilon_{3}^2$. By similar arguments above, we have $\chi_u(C_0(f))=\chi_u(C_1(f))=\chi_u(C_2(f)).$ From Fact \ref{bal} we also have $\chi_u(C_0(f))+\chi_u(C_1(f))+\chi_u(C_2(f))=0.$ Hence, we get $\chi_u(C_0(f))=\chi_u(C_1(f))=\chi_u(C_2(f))=0$. From Equation (\ref{wts}) we have $$wt (c_{u})=2\left(3^{r-2}-3^{\frac{n}{2}-2}+3^{\frac{n}{2}-1}\right).$$

	\end{itemize}

For $u \notin  \big{(}B_{+}(f)\big{)}^{\perp}$, let us now evaluate the weight distribution  of $\mathcal {C}_{C_{j_0}(f)}$.
\begin{itemize}

\item $u \in B_{+}(f^*)$ and $f(u)=j_0$ \\
 By Proposition \ref{dualsize},  $ |B_{+}(f^*)|=3^r$.  Since $f$ is of type\;$(+)$, from Propositions \ref{f*0} and \ref{dual}, we have   $$3^{\frac{n}{2}}\epsilon_{3}^{j_0}=\sum_{\alpha \in B_{+}(f^*)} \epsilon_{3}^{f(x)}.$$ By orthogonality relations of character sums we have $|C_{j_0}(f^*)|=3^{r-1}-3^{\frac{n}{2}-1}+3^{\frac{n}{2}}$ and $|C_{j}(f^*)|=3^{r-1}-3^{\frac{n}{2}-1}$ for $j\neq j_0 \in \F_3$ . From Lemma \ref{sonn}, the restriction of $f$ into the subset $ \big{(}B_{+}(f)\big{)}^{\perp}$ is the constant $j_0$. Since the size of $\Kernel(\theta)$ is $3^{n-r}$ then dividing $3^{r-1}-3^{\frac{n}{2}-1}+3^{\frac{n}{2}}-3^{n-r}$ by $3^{n-r}$ we get  $E_{w_1}=3^{2r-n-1}-3^{r-\frac{n}{2}-1}+3^{r-\frac{n}{2}}-1$, where $w_1=23^{r-2}$.
\item $u \in B_{+}(f^*)$ and $f(u)\neq j_0$ \\
From the previous part we have $E_{w_2}=23^{2r-n-1}-23^{r-\frac{n}{2}-1}$, where $w_2=23^{r-2}+23^{\frac{n}{2}-1}$.
\item $u \in B_{-}(f^*)$  \\
Dividing $|B_{-}(f^*)|$ by $3^{n-r}$ we get $E_{w_3}=3^r-3^{2r-n}$, where $_3=23^{r-2}-23^{\frac{n}{2}-2}+23^{\frac{n}{2}-1}$.
\end{itemize}
\end{proof}

	\begin{table}[!ht] 
		\begin{center}
			\begin{tabular}{|c|c| }
				\hline
				Hamming weight $a$  &  Multiplicity $E_a $  \\ 
				\hline  \hline  
				0  & 1 \\
				\hline
				$23^{r-2}$  &  $3^{2r-n-1}+3^{r-\frac{n}{2}-1}$ \\
				\hline
				$	2(3^{r-2}+3^{\frac{n}{2}-1})$ &$ 23^{2r-n-1}-23^{r-\frac{n}{2}-1} $\\
				\hline
				$2(3^{r-2}-3^{\frac{n}{2}-2}+3^{\frac{n}{2}-1})$  &$3^r-3^{2r-n}$\\
				\hline
			\end{tabular}
		\end{center}	 \caption{ \label{tabloo1}   The weight distribution of  $\mathcal {C}_{C_{j_0}(f)}$ when $n$ is even.}	
	\end{table}   

\par Next, we verify Theorem \ref{even} by MAGMA program for the following non-weakly regular bent functions. 

\begin{example} 
	
 Let $F:\mathbb{F}_{3}^{6}\simeq\mathbb{F}_{3}^{4}\times\mathbb{F}_{3}\times\mathbb{F}_{3}\rightarrow\mathbb{F}_{3}$ be a non-weakly regular bent function defined by Equation (\ref{gmmf}) where $F(x_1,x_2,x_3,x_4,x_5,x_6)=f^{(x_6)}(x_1,x_2,x_3,x_4)+x_5x_6 $. Let $f^{(x_6)}:\mathbb{F}_{3}^{4}\rightarrow\mathbb{F}_{3}$ be weakly regular bent of type ($+$) for $x_6=0$ and of type ($-$) for $x_6=1,$ and $x_6=2$, where $f^{(0)}(x_1,x_2,x_3,x_4)=2x_1^2+2x_2^2 + x_3^2 + x_4^2$, $f^{(1)}(x_1,x_2,x_3,x_4)=f^{(2)}(x_1,x_2,x_3,x_4)=x_1^2+x_2^2 + 2x_3^2 + x_4^2$. Then, by employing the Langrange interpolation technique given in \cite[Theorem 2]{Cesmelioglu2013b}, we have $F(x_1,x_2,x_3,x_4,x_5,x_6)=(-1)\big{(}f^{(0)}(x_1,x_2,x_3,x_4)(x_6-1)(x_6-2)+(f^{(1)}(x_1,x_2,x_3,x_4)
 +x_5)(x_6)(x_6-2)+(f^{(2)}(x_1,x_2,x_3,x_4)+2x_5)(x_6)(x_6-1) \big{)}=2x_1^2x_6^2 + 2x_1^2 + 2x_2^2x_6^2 + 2x_2^2 + x_3^2x_6^2 + x_3^2 + x_4^2 + x_5x_6$ is non-weakly regular of type ($+$), where $W^{+}(F)=\{0\}$ and $B_{+}(F)=\mathbb{F}_{3}^{4}\times\{0\}\times\mathbb{F}_{3}$. By \cite[Proposition 2]{Cesmelioglu2013}, we have ${f^{(0)}}^*(x_1,x_2,x_3,x_4)=x_1^2+x_2^2 + 2x_3^2 + 2x_4^2$,\; ${f^{(1)}}^*(x_1,x_2,x_3,x_4)={f^{(2)}}^*(x_1,x_2,x_3,x_4)=2x_1^2+2x_2^2 + x_3^2 + 2x_4^2$. Moreover, from Equation \eqref{gmmfdual} we know that $F^{*}$ belongs to GMMF class. Then, by employing the Langrange interpolation to the functions ${f^{(0)}}^*$, ${f^{(1)}}^*$, and ${f^{(2)}}^*,$ we get $F^{*}(x)=x_1^2x_5^2 + x_1^2 + x_2^2x_5^2 + x_2^2+ 2x_3^2x_5^2 + 2x_3^2 + 2x_4^2 + 2x_5x_6$.
	\begin{itemize}
		\item $m=4,\;s=1$,\; $F(x)=F(-x)$ and  $j_0=0$ ;
		\item $r=m+s+\Dim(W^{+}(F))=4+1+0=5$;
		\item The set $\mathcal {C}_{C_0(F)}$ is a three-weight linear code with parameters $[98,5,54]_3$,  weight enumerator $1+ 32y^{54} + 162y^{66} + 48y^{72} $.
		
	\end{itemize}
\end{example}

\begin{example}

	Let $F:\mathbb{F}_{3}^{6}\simeq\mathbb{F}_{3}^{4}\times\mathbb{F}_{3}\times\mathbb{F}_{3}\rightarrow\mathbb{F}_{3}$ be a non-weakly regular bent function defined by Equation (\ref{gmmf}), where $F(x_1,x_2,x_3,x_4,x_5,x_6)=f^{(x_6)}(x_1,x_2,x_3,x_4)+x_5x_6 $. Let $f^{(x_6)}:\mathbb{F}_{3}^{4}\rightarrow\mathbb{F}_{3}$ be weakly regular bent of type ($+$) for $x_6=0$ and of type ($-$) for $x_6=1,$ and $x_6=2,$ where $f^{(0)}(x_1,x_2,x_3,x_4)=x_1^2+2x_2^2 + 2x_3^2 + x_4^2+1$, $f^{(1)}(x_1,x_2,x_3,x_4)=f^{(2)}(x_1,x_2,x_3,x_4)=2x_1^2+2x_2^2 + 2x_3^2 + x_4^2$. Then, by employing the Langrange interpolation, we have $F(x_1,x_2,x_3,x_4,x_5,x_6)=(-1)\big{(}f^{(0)}(x_1,x_2,x_3,x_4)(x_6-1)(x_6-2)+(f^{(1)}(x_1,x_2,x_3,x_4)
	+x_5)(x_6)(x_6-2)+(f^{(2)}(x_1,x_2,x_3,x_4)+2x_5)(x_6)(x_6-1) \big{)}=x_1^2x_6^2 + x_1^2  + 2x_2^2 + 2x_3^2 + x_4^2 + x_5x_6+ 2x_6^2+1 $ is non-weakly regular of type ($+$), where $W^{+}(F)=\{0\}$ and $B_{+}(F)=\mathbb{F}_{3}^{4}\times\{0\}\times\mathbb{F}_{3}$. By \cite[Proposition 2]{Cesmelioglu2013}, we have ${f^{(0)}}^*(x_1,x_2,x_3,x_4)=2x_1^2+x_2^2 + x_3^2 + 2x_4^2$,\; ${f^{(1)}}^*(x_1,x_2,x_3,x_4)={f^{(2)}}^*(x_1,x_2,x_3,x_4)=x_1^2+x_2^2 + x_3^2 + 2x_4^2$. Moreover, from Equation \eqref{gmmfdual} we know that $F^{*}$ belongs to GMMF class. Then, by employing the Langrange interpolation to the functions ${f^{(0)}}^*$, ${f^{(1)}}^*$, and ${f^{(2)}}^*,$ we get $F^{*}(x)=2x_1^2x_5^2 + 2x_1^2  + x_2^2+ 2x_5^2 + x_3^2 + 2x_4^2 + 2x_5x_6+1$.
	\begin{itemize}
		\item $m=4,\;s=1$,\; $F(x)=F(-x)$ and  $j_0=1$ ;
		\item $r=m+s+\Dim(W^{+}(F))=4+1+0=5$;
		\item The set $\mathcal {C}_{C_1(F)}$ is a three-weight linear code with parameters $[98,5,54]_3$,  weight enumerator $1+ 32y^{54} + 162y^{66} + 48y^{72} $.
		
	\end{itemize}
\end{example}

Let $n$ be odd. Then, by Proposition \ref{dualtype} $f^*$ is of type $(-)$. We next determine the $|C_i(f)|$ for each $i\in \F_3$. Put $f(0)=j_0$. For any $u \in \mathbb{F}_{3}^{n}$, using Equations (\ref{inversee}) and (\ref{gauss}) we have \begin{equation} \label{oddsize} \begin{array}{ll}
 -i3^{\frac{n}{2}}\epsilon_{3}^{f(u)}&= 3^{\frac{n-1}{2}}\left(\epsilon_{3}^{f(u)+2}-\epsilon_{3}^{f(u)+1}\right)\\
&= \sum_{\alpha \in B_{+}(f)} \epsilon_{3}^{f^*(\alpha)+\alpha \cdot u}-\sum_{\alpha \in B_{-}(f)} \epsilon_{3}^{f^*(\alpha)+\alpha \cdot u}. \end{array}  \end{equation}  If $u=0$, then by Proposition \ref{dual} and Equation (\ref{oddsize}), we have $0=|C_{j_0}(f)|\epsilon_{3}^{j_0}+\left(|C_{j_0+1}(f)|+3^{\frac{n-1}{2}}\right)\epsilon_{3}^{j_0+1}+\left(|C_{j_0+2}(f)|-3^{\frac{n-1}{2}}\right)\epsilon_{3}^{j_0+2}$.
Hence, we have $|C_{j_0}(f)|=|C_{j_0+1}(f)|+3^{\frac{n-1}{2}}=|C_{j_0+2}(f)|-3^{\frac{n-1}{2}},$ where the summations $j_0+1$ and $j_0+2$ are modulo $3$. On the other hand,  $|C_0(f)|+|C_1(f)|+|C_2(f)|=3^r$. Hence, we get \be \label{car++} |C_{j_0+i}(f)|=3^{r-1}-\left( \frac{i}{3} \right)3^{\frac{n-1}{2}} \ee for each $i \in \F_3.$

\par For $f(0)=j_0$ let us take $S=C_{j_0+2}(f)$. Observe that the linear code $\sC_{C_{j_0+2}(f)}$ of length $3^{r-1}+3^{\frac{n-1}{2}}$ over $\F_3$ defined by $(\ref{defCode})$ is at most $r$-dimensional.

\begin{prp} 
	The linear code $\sC_{C_{j_0+2}(f)}$ of length $3^{r-1}+3^{\frac{n-1}{2}}$ over $\F_3$ defined  by $(\ref{defCode})$ is  $r$-dimensional.
\end{prp}

\begin{pf}
	The proof is similar to that of Proposition \ref{kernel} and is ommited here.
\end{pf}
\begin{rem}\label{rem4}
	Let $f:\mathbb{F}_{3}^{n}\rightarrow\mathbb{F}_{3}$  be a non-weakly regular dual-bent function with $f(0)=j_0$, $n$ be odd, and $B_{+}(f)$ be a vector space of dimension $r$ over $\F_3$. For some $i\in F_3$, let us consider the linear code $\sC_{C_{i}(f)}$ of length over $\F_3$ defined by $(\ref{defCode})$. As in Remark \ref{rem3}, to guarantee that $\Dim(\sC_{C_{i}(f)})=r$, we need to have $|C_{i}(f)^{\star}|>3^{r-1}$, and from Equation \eqref{car++} it is possible if and only if $i=j_0+2$.
\end{rem}

\begin{teo} \label{odd}
	
	Let $n$ be an odd integer and $f:\mathbb{F}_{3}^{n}\rightarrow\mathbb{F}_{3}$  be a non-weakly regular dual-bent function such that $f(0)=j_0$, and $f(x)=f(-x)$. Let  $B_{+}(f)$  be an $r$-dimensional non-degenerate $\mathbb {F}_{3}$-vector space with $r \ge \frac{n+1}{2}.$ Then, the code $\mathcal {C}_{C_{j_0+2}(f)}$ whose codewords $c_{u}$ are defined by Equation (\ref{defCode}) is a three weight ternary linear code with parameters $[3^{r-1}+3^{\frac{n-1}{2}},r,23^{r-2}]_{3}$. The codeword $c_{u}$ has zero-weight if $u \in  \big{(}B_{+}(f)\big{)}^{\perp}$ i.e., $u \in \Kernel(\theta)$. For $u \notin  \big{(}B_{+}(f)\big{)}^{\perp}$ the non-zero weight codewords are as follows.
	
	$$ wt (c_{u})=\left\{ \begin{array}{lll}
	23^{r-2}& \mbox{if $u \in B_{-}(f^*)$ and $f(u)=j_0$};\\
	2\left(3^{r-2}+3^{\frac{n-3}{2}}\right) & \mbox{if $u \in B_{-}(f^*)$ and $f(u)=j_0+2$ or $u \in B_{+}(f^*)$ };\\
	2\left(3^{r-2}+23^{\frac{n-3}{2}}\right) & \mbox{if $u \in B_{-}(f^*)$ and $f(u)=j_0+1$ }.
	\end{array}
	\right. $$
	Moreover,  the weight distribution of $\mathcal {C}_{C_{j_0+2}(f)}$ is as in Table \ref{tabloo2}.
\end{teo}

\begin{proof}We now evaluate the non-zero weights of codewords in $\sC_{C_{j_0+2}(f)}$. For $c_{u}\in \sC_{C_{j_0+2}(f)}$, we have the following.

	\begin{itemize}
		\item $u \in B_{+}(f)^{\perp}$
		\\
		By Equation \eqref{preim}, we have $ C_{j_0+2}(f) \subset B_{+}(f)$ which implies $wt(c_{u})=0.$
		\item   $u \notin B_{+}(f)^{\perp} $ and $u \in B_{-}(f^*)$
		\\
		From Proposition \ref{dual} and Equation (\ref{oddsize}) we have  $$3^{\frac{n-1}{2}}\left(\epsilon_{3}^{f(u)+2}-\epsilon_{3}^{f(u)+1}\right)=\chi_u(C_0(f))+\chi_u(C_1(f))\epsilon_{3}+\chi_u(C_2(f))\epsilon_{3}^2.$$ Let $f(u)=j.$ Then, we have $\chi_u(C_j(f))=\chi_u(C_{j+1}(f))+3^{\frac{n-1}{2}}=\chi_u(C_{j+2}(f)) -3^{\frac{n-1}{2}},$ where the summations $j+1$ and $j+2$ are modulo $3$. By Fact \ref{bal}, we also have $\chi_u(C_0(f))+\chi_u(C_1(f))+\chi_u(C_2(f))=0.$ Hence, we get $\chi_u(C_j(f))=0$, $\chi_u(C_{j+1}(f))=-3^{\frac{n-1}{2}}$,\;$\chi_u(C_{j+2}(f))=3^{\frac{n-1}{2}}$. From Equation (\ref{wts}) we have $$wt(c_{u}) \in \{23^{r-2},2\left(3^{r-2}+3^{\frac{n-3}{2}}\right),2\left(3^{r-2}+23^{\frac{n-3}{2}}\right)\}.$$
		
		\item  $u \notin B_{+}(f)^\perp $ and $u \in B_{+}(f^*).$
		\\
		From Proposition \ref{dual} we have $0=\chi_u(C_0(f))+\chi_u(C_1(f))\epsilon_{3}+\chi_u(C_2(f))\epsilon_{3}^2$. Then, we have $\chi_u(C_0(f))=\chi_u(C_1(f))=\chi_u(C_2(f)).$ From Fact \ref{bal} we also have $\chi_u(C_0(f))+\chi_u(C_1(f))+\chi_u(C_2(f))=0.$ Hence, we get $\chi_u(C_0(f))=\chi_u(C_1(f))=\chi_u(C_2(f))=0$. From Equation (\ref{wts}) we have $$wt (c_{u})=2\left(3^{r-2}+3^{\frac{n-3}{2}}\right).$$
		
	\end{itemize}

For $u \notin  \big{(}B_{+}(f)\big{)}^{\perp}$, let us now evaluate the weight distribution of $\mathcal {C}_{C_{j_0+2}(f)}$.
\begin{itemize}

\item $u \in B_{-}(f^*)$ and $f(u)=j_0$ \\
 By Proposition \ref{dualsize},  $ |B_{-}(f^*)|=3^r$. 
Since $f$ is of type\;$(+)$, from Equation \ref{gauss} and Propositions \ref{f*0}, \ref{dual}, we have   $$i3^{\frac{n}{2}}\epsilon_{3}^{j_0}=3^{\frac{n-1}{2}}\sum_{j\in\mathbb{F}_3^{\star}}\left( \frac{j}{3} \right)\epsilon_3^{j_0+j}=\sum_{\alpha \in B_{-}(f^*)} \epsilon_{3}^{f(x)}.$$ By orthogonality relations of character sums we have $|D_{j_0}(f^*)|=3^{r-1}$ and $ |D_{j+j_0}(f^*)|=3^{r-1}+\left( \frac{j}{3} \right)3^{\frac{n-1}{2}})$ for all $j\neq 0 \in F_3$. From Lemma \ref{sonn},  the restriction of $f$ into the subset $ \big{(}B_{+}(f)\big{)}^{\perp}$ is the constant $j_0$. Dividing $3^{r-1}-3^{n-r}$ by $3^{n-r}$ we get  $E_{w_1}=3^{2r-n-1}-1$, where $w_1=23^{r-2}$.
\item $u \in B_{-}(f^*)$ and $f(u)=j_0+2$ or $u \in B_{+}(f^*)$ \\
From the arguments above we have $E_{w_2}=3^{r}-23^{2r-n-1}-3^{r-\frac{n+1}{2}}$, where $w_2=23^{r-2}+23^{\frac{n-3}{2}}$.
\item $u \in B_{-}(f^*)$  and $f(u)=j_0+1$ \\
From the arguments above we have $E_{w_3}=3^{2r-n-1}+3^{r-\frac{n+1}{2}}$, where $w_3=23^{r-2}+43^{\frac{n-3}{2}}$.

\end{itemize}
\end{proof}

	\begin{table}[!ht] 
		\begin{center}
			\begin{tabular}{|c|c| }
				\hline
				Hamming weight $a$  &  Multiplicity $E_a $  \\ 
				\hline  \hline  
				0  & 1 \\
				\hline
				$23^{r-2}$  &  $3^{2r-n-1}-1$ \\
				\hline
				$	2\left(3^{r-2}+3^{\frac{n-3}{2}}\right)$ &$ 3^{r}-23^{2r-n-1}-3^{r-\frac{n+1}{2}} $\\
				\hline
				$2\left(3^{r-2}+23^{\frac{n-3}{2}}\right)$  &$3^{2r-n-1}+3^{r-\frac{n+1}{2}}$\\
				\hline
			\end{tabular}
		\end{center}	 \caption{ \label{tabloo2}   The weight distribution of  $\mathcal {C}_{C_{j_0+2}(f)}$ when $n$ is odd.}	
	\end{table}

\par Next, we verify Theorem \ref{odd} by MAGMA program for the following non-weakly regular bent functions.

\begin{example} \label{ex12} 
	
	Let $F:\mathbb{F}_{3}^{7}\simeq\mathbb{F}_{3}^{5}\times\mathbb{F}_{3}\times\mathbb{F}_{3}\rightarrow\mathbb{F}_{3}$ be a non-weakly regular bent function defined by Equation (\ref{gmmf}), where $F(x_1,x_2,x_3,x_4,x_5,x_6,x_7)=f^{(x_7)}(x_1,x_2,x_3,x_4,x_5)+x_6x_7 $. Let $f^{(x_7)}:\mathbb{F}_{3}^{5}\rightarrow\mathbb{F}_{3}$ be weakly regular bent of type ($+$) for $x_7=0$ and of type ($-$) for $x_7=1,$ and $x_7=2,$ where $f^{(0)}(x_1,x_2,x_3,x_4)=2x_1^2+x_2^2 + 2x_3^2 + x_4^2+ x_5^2$, $f^{(1)}(x_1,x_2,x_3,x_4)=f^{(2)}(x_1,x_2,x_3,x_4)=x_1^2+x_2^2 + x_3^2 + x_4^2+ 2x_5^2$. Then, by employing the Langrange interpolation, we have $F(x_1,x_2,x_3,x_4,x_5,x_6,x_7)=(-1)\big{(}f^{(0)}(x_1,x_2,x_3,x_4,x_5)(x_7-1)(x_7-2)+(f^{(1)}(x_1,x_2,x_3,x_4,x_5)+x_6)(x_7)(x_7-2)\\
	+(f^{(2)}(x_1,x_2,x_3,x_4,x_5)+2x_6)(x_7)(x_7-1) \big{)}=2x_1^2x_7^2 + 2x_1^2 + x_2^2 + 2x_3^2x_7^2+ 2x_3^2 + x_4^2 + x_5^2x_7^2+ +x_5^2 + x_6x_7$ is non-weakly regular of type ($+$), where $W^{+}(F)=\{0\}$ and $B_{+}(F)=\mathbb{F}_{3}^{5}\times\{0\}\times\mathbb{F}_{3}$. By \cite[Proposition 2]{Cesmelioglu2013}, we have ${f^{(0)}}^*(x_1,x_2,x_3,x_4,x_5)=x_1^2+2x_2^2 + x_3^2 + 2x_4^2+ 2x_5^2$,\; ${f^{(1)}}^*(x_1,x_2,x_3,x_4,x_5)={f^{(2)}}^*(x_1,x_2,x_3,x_4,x_5)=2x_1^2+2x_2^2 + 2x_3^2 + 2x_4^2 + x_5^2$. Moreover, from Equation \eqref{gmmfdual} we know that $F^{*}$ belongs to GMMF class. Then, by employing the Langrange interpolation to the functions ${f^{(0)}}^*$, ${f^{(1)}}^*$, and ${f^{(2)}}^*,$ we get $F^{*}(x)=x_1^2x_6^2 + x_1^2  + 2x_2^2+ x_3^2x_6^2 + x_3^2 + 2x_4^2 + 2x_5^2x_6^2+2x_5^2+ 2x_6x_7$.

	\begin{itemize}
		\item $m=5,\;s=1$,\; $F(x)=F(-x)$ and  $j_0=0$ ;
		\item $r=m+s+\Dim(W^{+}(F))=5+1+0=6$;
		\item The set $\mathcal {C}_{C_2(F)}$ is a three-weight linear code with parameters $[270,6,162]_3$,  weight enumerator $1+ 80y^{162} + 558y^{180} + 90y^{198}$.
		
	\end{itemize}
\end{example}

\begin{example} \label{ex13} 
	
	Let $F:\mathbb{F}_{3}^{7}\simeq\mathbb{F}_{3}^{5}\times\mathbb{F}_{3}\times\mathbb{F}_{3}\rightarrow\mathbb{F}_{3}$ be a non-weakly regular bent function defined by Equation (\ref{gmmf}), where $F(x_1,x_2,x_3,x_4,x_5,x_6,x_7)=f^{(x_6,x_7)}(x_1,x_2,x_3)+x_4x_6+x_5x_7 $. Let $f^{(x_6,x_7)}:\mathbb{F}_{3}^{3}\rightarrow\mathbb{F}_{3}$ be weakly regular bent of type ($+$) for $(x_6,x_7) \in \{(0,0),(1,1),(2,2)\}$ and of type ($-$) for $(x_6,x_7) \in \{(0,1),(0,2),(1,0),(2,0),(1,2),\\
	(2,1)\},$ where $f^{(0,0)}(x_1,x_2,x_3)=2x_1^2+x_2^2 + x_3^2 + 2$,\; $f^{(1,1)}(x_1,x_2,x_3)=f^{(2,2)}(x_1,x_2,x_3)\\
	=x_1^2+2x_2^2 + x_3^2 $,\;$f^{(0,1)}(x_1,x_2,x_3)=f^{(0,2)}(x_1,x_2,x_3)=x_1^2+x_2^2 + x_3^2 $,\;$f^{(1,0)}(x_1,x_2,x_3)=f^{(2,0)}(x_1,x_2,x_3)=2x_1^2+2x_2^2 + x_3^2 $,\; $f^{(1,2)}(x_1,x_2,x_3)=f^{(2,1)}(x_1,x_2,x_3)=2x_1^2+x_2^2 + 2x_3^2 $. Then, by employing the Langrange interpolation, we have $F(x_1,x_2,x_3,x_4,x_5,x_6 \\
	,x_7)=\big{(}f^{(0,0)}(x_1,x_2,x_3)(x_6-1)(x_6-2)(x_7-1)(x_7-2)+(f^{(1,1)}(x_1,x_2,x_3)+x_4+x_5)(x_6)(x_6-2)(x_7)(x_7-2)+(f^{(2,2)}(x_1,x_2,x_3)+2x_4+2x_5)(x_6)(x_6-1)(x_7)(x_7-1)+(f^{(0,1)}(x_1,x_2,x_3)+x_5)(x_6-1)(x_6-2)(x_7)(x_7-2)+(f^{(0,2)}(x_1,x_2,x_3)+2x_5)(x_6-1)(x_6-2)(x_7)(x_7-1)  +(f^{(1,0)}(x_1,x_2,x_3)+x_4)(x_6)(x_6-2)(x_7-1)(x_7-2) +(f^{(2,0)}(x_1,x_2,x_3)+2x_4)(x_6)(x_6-1)(x_7-1)(x_7-2) +(f^{(1,2)}(x_1,x_2,x_3)+x_4+2x_5)(x_6)(x_6-2)(x_7)(x_7-1) +(f^{(2,1)}(x_1,x_2,x_3)+2x_4+x_5)(x_6)(x_6-1)(x_7)(x_7-2)\big{)}=2x_1^2x_6^2x_7^2 + x_1^2x_6x_7  + 2x_1^2x_7^2+2x_1^2+x_2^2x_6^2x_7^2+x_2^2x_6^2+2x_2^2x_6x_7+x_2^2+2x_3^2x_6^2x_7^2 + x_3^2x_6x_7+x_3^2+2x_6^2x_7^2+x_6^2+ x_7^2+ x_4x_6+ x_5x_7+2$ is non-weakly regular of type ($+$), where $W^{+}(F)=\{(0,0),(1,1),(2,2)\}$ and $B_{+}(F)=\mathbb{F}_{3}^{5}\times\{(0,0),(1,1),(2,2)\}\times\mathbb{F}_{3}$. By \cite[Proposition 2]{Cesmelioglu2013}, we have ${f^{(0,0)}}^*(x_1,x_2,x_3)=x_1^2+2x_2^2 + 2x_3^2+2$,\; ${f^{(1,1)}}^*(x_1,x_2,x_3)={f^{(2,2)}}^*(x_1,x_2,x_3)=2x_1^2+x_2^2 + 2x_3^2$, \; ${f^{(0,1)}}^*(x_1,x_2,x_3)={f^{(0,2)}}^*(x_1,x_2,x_3)=2x_1^2+2x_2^2 + 2x_3^2$, \; ${f^{(1,0)}}^*(x_1,x_2,x_3)={f^{(2,0)}}^*(x_1,x_2,x_3)=x_1^2+x_2^2 + 2x_3^2$, \; ${f^{(1,2)}}^*(x_1,x_2,x_3)={f^{(2,1)}}^*(x_1,x_2,x_3)=x_1^2+2x_2^2 + x_3^2$. Moreover, from Equation \eqref{gmmfdual} we know that $F^{*}$ belongs to GMMF class. Then, by employing the Langrange interpolation to the functions ${f^{(i,j)}}^*$, where $(i,j) \in \F_3^2$, we get 
	$F^{*}(x)=x_1^2x_4^2x_5^2 + 2x_1^2x_4x_5+x_1^2x_5^2+x_1^2 +2x_2^2x_4^2x_5^2+2x_2^2x_4^2+2x_2^2x_4x_5 + 2x_2^2+ x_3^2x_4^2x_5^2 + 2x_3^2x_4x_5+2x_4^2x_5^2 + 2x_3^2 + x_4^2 + 2x_4x_6+x_5^2+ 2x_5x_7+2$.

	\begin{itemize}
		\item $m=3,\;s=2$,\; $F(x)=F(-x)$ and  $j_0=2$ ;
		\item $r=m+s+\Dim(W^{+}(F))=3+2+1=6$;
		\item The set $\mathcal {C}_{C_1(f)}$ is a three-weight linear code with parameters $[270,6,162]_3$,  weight enumerator $1+ 80y^{162} + 558y^{180} + 90y^{198}$.
		
	\end{itemize}
\end{example}

\section{Three-weight linear codes on $\mathbf{B_{-}(f)}$}\label{SectionConstruction-}

\par In Section \ref{SectionConstruction-}, we assume $B_{-}(f)$ is an $\mathbb {F}_{3}$-vector space with $\dim(B_{-}(f))\ge \lfloor \frac{n}{2}\rfloor+1.$ Put $\dim(B_{-}(f))=r$. Observe that $0 \in B_{-}(f)$ implies $f$ is of type $(-)$. Let $n$ be even. Then, by Proposition \ref{dualtype} $f^*$ is also of type $(-)$. We next determine the $|D_i(f)|$ for each $i\in \F_3$. Put $f(0)=j_0$. Then, by Proposition \ref{dual} we have $-3^{\frac{n}{2}}\epsilon_{3}^{j_0}=|D_0(f)|+|D_1(f)|\epsilon_{3}+|D_2(f)|\epsilon_{3}^2$. Then, by similar arguments as in Section \ref{SectionConstruction}, we have $|D_{j_0}(f)|+3^{\frac{n}{2}}=|D_{j_0+1}(f)|=|D_{j_0+2}(f)|.$ On the other hand, we have $|D_0(f)|+|D_1(f)|+|D_2(f)|=3^r$. Hence, we get \be \label{car-}|D_{j_0}(f)|=3^{r-1}+3^{\frac{n}{2}-1}-3^{\frac{n}{2}},\;\; |D_{j_0+1}(f)|=|D_{j_0+2}(f)|=3^{r-1}+3^{\frac{n}{2}-1}.\ee

\par Let us take $S=D_{j_0+2}(f)$. Observe that the linear code $\sC_{D_{j_0+2}(f)}$ of length $3^{r-1}+3^{\frac{n}{2}-1}$ over $\F_3$ defined by $(\ref{defCode})$ is at most $r$-dimensional.

\begin{prp}\label{ker1}
	The linear code $\sC_{D_{j_0+2}(f)}$ of length $3^{r-1}+3^{\frac{n}{2}-1}$ over $\F_3$ defined by (\ref{defCode}) is a $r$-dimensional subspace of $\mathbb{F}_{3}^{3^{r-1}+3^{\frac{n}{2}-1}}$.
\end{prp}

\begin{proof}
	The proof is similar to that of  Proposition \ref{kernel} and is ommited here.
\end{proof}

	\begin{rem}\label{rem5}
		Let $f:\mathbb{F}_{3}^{n}\rightarrow\mathbb{F}_{3}$  be a non-weakly regular dual-bent function with $f(0)=j_0$, $n$ be even, and $B_{-}(f)$ be a vector space of dimension $r$ over $\F_3$. For some $i\in F_3$, let us consider the linear code $\sC_{D_{i}(f)}$ over $\F_3$ defined by $(\ref{defCode})$. As in Remark \ref{rem3}, to guarantee that $\Dim(\sC_{D_{i}(f)})=r$ we need to have $|D_{i}(f)^{\star}|>3^{r-1}$ and from Equation \eqref{car-} it is possible if and only if $i\ne j_0$.
	\end{rem}

Note that we can also set $S=D_{j_0+1}(f)$ so that Proposition \ref{ker1} holds. Since, it requires similar calculations we only consider one of them to avoid repetition.
Since $B_{-}(f)$ is a vector space and $f^*(x)=f^*(-x)$, then we have $D_i(f)=-D_i(f)$ for all $i \in \F_3$ which implies $\chi_u(D_i(f)) \in \Z$ for all $i\in \F_3$, $u\in \F_3^n$. Then, by Equation (\ref{wt}), for all $i\in \F_3$ we have \be \label{wts-} wt(c_{u})=\frac{2}{3}\left(|D_i(f)^{\star}|-\chi_u(D_i(f)^{\star})\right).\ee

\begin{teo} \label{even-}
	
	Let $n$ be an even integer, $f:\mathbb{F}_{3}^{n}\rightarrow\mathbb{F}_{3}$  be a non-weakly regular dual-bent function such that $f(0)=j_0$ and $f(x)=f(-x)$. Let  $B_{-}(f)$  be an $r$-dimensional non-degenerate $\mathbb {F}_{3}$-vector space with $r\ge \frac{n}{2}+1.$ Then, the code $\mathcal {C}_{D_{j_0+2}(f)}$ whose codewords $c_{u}$ are defined by Equation (\ref{defCode}) is a three weight ternary linear code with parameters $[3^{r-1}+3^{\frac{n}{2}-1},r,23^{r-2}]_{3}$. The codeword $c_{u}$ has zero-weight if $u \in  \big{(}B_{-}(f)\big{)}^{\perp}$ i.e., $u \in \Kernel(\theta)$. For $u \notin  \big{(}B_{-}(f)\big{)}^{\perp}$ the non-zero weight codewords are as follows.
	
	$$ wt (c_{u})=\left\{ \begin{array}{lll}
	23^{r-2}& \mbox{if $u \in B_{-}(f^*)$ and $f(u) \ne j_0+2$};\\
	2\left(3^{r-2}+3^{\frac{n}{2}-1}\right) & \mbox{if $u \in B_{-}(f^*)$ and $f(u)=j_0+2$ };\\
	2\left(3^{r-2}+3^{\frac{n}{2}-2}\right)& \mbox{if $u \in B_{+}(f^*)$}.
	\end{array}
	\right. $$
	Moreover,  the weight distribution of $\mathcal {C}_{D_{j_0+2}(f)}$ is as in Table \ref{tabloo3}.
\end{teo}

\begin{proof}
We now evaluate the non-zero weights of codewords in $\sC_{D_{j_0+2}(f)}$. For $c_{u}\in \sC_{D_{j_0+2}(f)}$, we have the following.

	\begin{itemize}
		\item $u \in B_{-}(f)^{\perp}$
		\\
		By Equation \eqref{preim}, we have $ D_{j_0+2}(f) \subset B_{-}(f)$ which implies $wt(c_{u})=0.$
		
		\item $u \notin B_{-}(f)^{\perp} $ and $u \in B_{-}(f^*).$
		\\
		From Proposition \ref{dual} we have $-3^{\frac{n}{2}}\epsilon_{3}^{f(u)}=\chi_u(D_0(f))+\chi_u(D_1(f))\epsilon_{3}+\chi_u(D_2(f))\epsilon_{3}^2$. Let $f(u)=j.$ Then we have $\chi_u(D_j(f))+3^{\frac{n}{2}}=\chi_u(D_{j+1}(f))=\chi_u(D_{j+2}(f)),$ where the summations $j+1$ and $j+2$ are modulo $3$. From Fact \ref{bal} we also have $\chi_u(D_0(f))+\chi_u(D_1(f))+\chi_u(D_2(f))=0.$ Hence, we have $\chi_u(D_j(f))=-23^{\frac{n}{2}-1}$ and $\chi_u(D_{j+1}(f))=\chi_u(D_{j+2}(f))=3^{\frac{n}{2}-1}$. From Equation (\ref{wts-}) we have $$wt(c_{u}) \in \{23^{r-2},2(3^{r-2}+3^{\frac{n}{2}-1})\}.$$
		
		\item  $u \notin B_{-}(f)^{\perp} $ and $u \in B_{+}(f^*).$
		\\
		From Proposition \ref{dual} we have $0=\chi_u(D_0(f))+\chi_u(D_1(f))\epsilon_{3}+\chi_u(D_2(f))\epsilon_{3}^2$. Then we have $\chi_u(D_0(f))=\chi_u(D_1(f))=\chi_u(D_2(f)).$ From Fact \ref{bal} we also have $\chi_u(D_0(f))+\chi_u(D_1(f))+\chi_u(D_2(f))=0.$ Hence, we get $\chi_u(D_0(f))=\chi_u(D_1(f))=\chi_u(D_2(f))=0$. From Equation (\ref{wts-}) we have $$wt (c_{u})=2\left(3^{r-2}+3^{\frac{n}{2}-2}\right).$$

	\end{itemize}

For $u \notin  \big{(}B_{-}(f)\big{)}^{\perp}$, let us now evaluate the weight distribution of $\mathcal {C}_{D_{j_0+2}(f)}$.
\begin{itemize}

\item  $u \in B_{-}(f^*)$ and $f(u) \ne j_0+2$ \\
 By Proposition \ref{dualsize}, $ |B_{-}(f^*)|=3^r$.  Since $f$ is of type\;$(-)$, from Propositions \ref{f*0} and \ref{dual}, we have   $$-3^{\frac{n}{2}}\epsilon_{3}^{j_0}=\sum_{\alpha \in B_{-}(f^*)} \epsilon_{3}^{f(x)}.$$ By orthogonality relations of character sums we have $|D_{j_0}(f^*)|=3^{r-1}+3^{\frac{n}{2}-1}-3^{\frac{n}{2}}$ and $|D_{j}(f^*)|=3^{r-1}+3^{\frac{n}{2}-1}$ for $j\neq j_0 \in \F_3$. Then, $|D_{j_0}(f^*)|+|D_{j_0+1}(f^*)|=23^{r-1}+23^{\frac{n}{2}-1}-3^{\frac{n}{2}}$ From Lemma \ref{sonn},  the restriction of $f$ into the subset $ \big{(}B_{-}(f)\big{)}^{\perp}$ is the constant $j_0$. Dividing $23^{r-1}-3^{\frac{n}{2}-1}-3^{n-r}$ by $3^{n-r}$ we get  $E_{w_1}=23^{2r-n-1}-3^{r-\frac{n}{2}-1}-1$, where $w_1=23^{r-2}$.
\item $u \in B_{-}(f^*)$ and $f(u)=j_0+2$ \\
Since, $|D_{j_0+2}(f^*)|=3^{r-1}+3^{\frac{n}{2}-1}$ then we have $E_{w_2}=3^{2r-n-1}+3^{r-\frac{n}{2}-1}$, where $w_2=23^{r-2}+23^{\frac{n}{2}-1}$.
\item $u \in B_{+}(f^*)$  \\
Dividing $|B_{+}(f^*)|$ by $3^{n-r}$ we get $E_{w_3}=3^r-3^{2r-n}$, where $w_3=23^{r-2}+23^{\frac{n}{2}-2}$.

\end{itemize}
\end{proof}

	\begin{table}[!ht] 
		\begin{center}
			\begin{tabular}{|c|c| }
				\hline
				Hamming weight $a$  &  Multiplicity $E_a $  \\ 
				\hline  \hline  
				0  & 1 \\
				\hline
				$23^{r-2}$  &  $23^{2r-n-1}-3^{r-\frac{n}{2}-1}-1$ \\
				\hline
				$2\left(3^{r-2}+3^{\frac{n}{2}-1}\right)$ & $ 3^{2r-n-1}+3^{r-\frac{n}{2}-1} $\\
				\hline
				$ 2\left(3^{r-2}+3^{\frac{n}{2}-2}\right)$  & $3^r-3^{2r-n}$\\
				\hline
			\end{tabular}
		\end{center}	 \caption{ \label{tabloo3}   The weight distribution of  $\mathcal {C}_{D_{j_0+2}(f)}$  when $n$ is even.}	
	\end{table}

\par Next, we verify Theorem \ref{even-} by MAGMA program for the following ternary non-weakly regular bent function.
\begin{example}	
	Let $F:\mathbb{F}_{3}^{8}\simeq\mathbb{F}_{3}^{6}\times\mathbb{F}_{3}\times\mathbb{F}_{3}\rightarrow\mathbb{F}_{3}$ be a non-weakly regular bent function defined by Equation (\ref{gmmf}), where $F(x_1,x_2,x_3,x_4,x_5,x_6,x_7,x_8)=f^{(x_8)}(x_1,x_2,x_3,x_4,x_5,\\
	x_6)+x_7x_8 $. Let $f^{(x_8)}:\mathbb{F}_{3}^{6}\rightarrow\mathbb{F}_{3}$ be weakly regular bent of type ($-$) for $x_8=0$ and of type ($+$) for $x_8=1,$ and $x_8=2$, where $f^{(0)}(x_1,x_2,x_3,x_4,x_5,x_6)=2x_1^2+2x_2^2 + x_3^2 + x_4^2+x_5^2+x_6^2$, $f^{(1)}(x_1,x_2,x_3,x_4,x_5,x_6)=f^{(2)}(x_1,x_2,x_3,x_4,x_5,x_6)=x_1^2+2x_2^2 + 2x_3^2 + 2x_4^2+x_5^2+x_6^2$. Then, by employing Langrange interpolation, we have $F(x_1,x_2,x_3,x_4,x_5,x_6,x_7,x_8)=(-1)\big{(}f^{(0)}(x_1,x_2,x_3,x_4,x_5,x_6)(x_8-1)(x_8-2)+(f^{(1)}(x_1,x_2,x_3,x_4,x_5,x_6)
	+x_7)(x_8)(x_8-2)+(f^{(2)}(x_1,x_2,x_3,x_4,x_5,x_6)+2x_7)(x_8)(x_8-1) \big{)}=2x_1^2x_8^2 + 2x_1^2 + x_4^2x_8^2 + 2x_2^2 + x_3^2x_8^2 + x_3^2 + x_4^2+ x_5^2 + x_6^2 + x_7x_8$ is non-weakly regular of type ($-$), where $W^{-}(F)=\{0\}$ and $B_{-}(F)=\mathbb{F}_{3}^{6}\times\{0\}\times\mathbb{F}_{3}$. By \cite[Proposition 2]{Cesmelioglu2013}, we have ${f^{(0)}}^*(x_1,x_2,x_3,x_4,x_5,x_6)=x_1^2+x_2^2 + 2x_3^2 + 2x_4^2+2x_5^2+2x_6^2$,\; ${f^{(1)}}^*(x_1,x_2,x_3,x_4,x_5,x_6)={f^{(2)}}^*(x_1,x_2,x_3,x_4)=2x_1^2+x_2^2 + x_3^2 + x_4^2+2x_5^2+2x_6^2$. Moreover, from Equation \eqref{gmmfdual} we know that $F^{*}$ belongs to GMMF class. Then, by employing the Langrange interpolation to the functions ${f^{(0)}}^*$, ${f^{(1)}}^*$, and ${f^{(2)}}^*$, we get $F^{*}(x)=x_1^2x_7^2 + x_1^2 + 2x_4^2x_7^2 + x_2^2+ 2x_3^2x_7^2 + 2x_3^2 + 2x_4^2+ 2x_5^2 + 2x_6^2 + 2x_7x_8$.

	\begin{itemize}
		\item $m=6,\;s=1$,\; $F(x)=F(-x)$ and  $j_0=0$ ;
		\item $r=m+s+\Dim(W^{+}(F))=6+1+0=7$;
		\item The set $\mathcal {C}_{D_2(F)}$ is a three-weight linear code with parameters $[756,7,486]_3$,  weight enumerator $1+ 476y^{486} + 1458y^{504} + 252y^{540} $.
		
	\end{itemize}
\end{example}

Let $n$ be odd. Then, by Proposition \ref{dualtype} $f^*$ is of type $(+)$.  We next determine the $|D_i(f)|$ for each $i\in \F_3$. Put $f(0)=j_0$. For any $u \in \mathbb{F}_{3}^{n}$ using Equations (\ref{inversee}) and (\ref{gauss}) we have \begin{equation} \label{oddsize-} \begin{array}{ll}

 -i3^{\frac{n}{2}}\epsilon_{3}^{f(u)}&= 3^{\frac{n-1}{2}}\left(\epsilon_{3}^{f(u)+2}-\epsilon_{3}^{f(u)+1}\right)\\
 &=\sum_{\alpha \in B_{+}(f)} \epsilon_{3}^{f^*(\alpha)+\alpha \cdot u}-\sum_{\alpha \in B_{-}(f)} \epsilon_{3}^{f^*(\alpha)+\alpha \cdot u}. \end{array} \end{equation}  If $u=0$, then by Proposition \ref{dual} and Equation (\ref{oddsize-}), we have $0=|D_{j_0}(f)|\epsilon_{3}^{j_0}+\left(|D_{j_0+1}(f)|-3^{\frac{n-1}{2}}\right)\epsilon_{3}^{j_0+1}+\left(|D_{j_0+2}(f)|+3^{\frac{n-1}{2}}\right)\epsilon_{3}^{j_0+2}$.
Then, we have $|D_{j_0}(f)|=|D_{j_0+1}|-3^{\frac{n-1}{2}}=|D_{j_0+2}|+3^{\frac{n-1}{2}}.$ On the other hand, we have $|D_0(f)|+|D_1(f)|+|D_2(f)|=3^r$. Hence, we get \be \label{car--} |D_{j_0+i}(f)|=3^{r-1}+\left( \frac{i}{3} \right)3^{\frac{n-1}{2}} \ee for each $i \in \F_3.$

\par Let us take $S=D_{j_0+1}(f)$. Observe that the linear code $\sC_{D_{j_0+1}(f)}$ of length $3^{r-1}+3^{\frac{n-1}{2}}$ over $\F_3$ defined by $(\ref{defCode})$ is at most $r$-dimensional.

\begin{prp} 
	The linear code $\sC_{D_{j_0+1}(f)}$ of length $3^{r-1}+3^{\frac{n-1}{2}}$ over $\F_3$ defined  by $(\ref{defCode})$ is  $r$-dimensional.
\end{prp}

\begin{pf}
	The proof is similar to that of  Proposition \ref{kernel} and is ommited here.
\end{pf}

	\begin{rem}\label{rem6}
	Let $f:\mathbb{F}_{3}^{n}\rightarrow\mathbb{F}_{3}$  be a non-weakly regular dual-bent function with $f(0)=j_0$, $n$ be odd, and $B_{-}(f)$ be a vector space of dimension $r$ over $\F_3$. For some $i\in F_3$, let us consider the linear code $\sC_{D_{i}(f)}$ of length over $\F_3$ defined by $(\ref{defCode})$. By similar arguments as in that of Remark \ref{rem3}, to guarantee that $\Dim(\sC_{D_{i}(f)})=r$, we need to have $|D_{i}(f)|>3^{r-1}$, and from Equation \eqref{car--} it is possible if and only if $i = j_0+1$.
\end{rem}

\begin{con}\label{odd-}

	Let $n$ be an odd integer and $f:\mathbb{F}_{3}^{n}\rightarrow\mathbb{F}_{3}$ be a non-weakly regular dual-bent function such that $f(0)=0$, and $f(x)=f(-x)$. Let  $B_{-}(f)$  be an $r$-dimensional non-degenerate $\mathbb {F}_{3}$-vector space with $r \ge \frac{n+1}{2}.$ 
\end{con}
Let $f$ satisfies the Condition \ref{odd-}. We now evaluate the non-zero weights of codewords in $\sC_{D_{j_0+1}(f)}$. For $c_{u}\in \sC_{D_{j_0+1}(f)}$ we have the following.

	\begin{itemize}
		\item $u \in B_{-}(f)^{\perp}$\\
		By Equation \eqref{preim}, we have $ D_{j_0+1}(f) \subset B_{-}(f)$ which implies $wt(c_{u})=0.$
		\item   $u \notin {B_{-}(f)}^{\perp} $ and $u \in B_{+}(f^*).$\\
From Proposition \ref{dual} and Equation (\ref{oddsize-}) we have $3^{\frac{n-1}{2}}\left(\epsilon_{3}^{f(u)+1}-\epsilon_{3}^{f(u)+2}\right)=\chi_u(D_0(f))+\chi_u(D_1(f))\epsilon_{3} +\chi_u(D_2(f))\epsilon_{3}^2$. Let $f(u)=j.$ Then, we have $\chi_u(D_j(f))=\chi_u(D_{j+1}(f))-3^{\frac{n-1}{2}}=\chi_u(D_{j+2}(f))+3^{\frac{n-1}{2}},$ where the summations $j+1$ and $j+2$ are modulo $3$. By Fact \ref{bal}, we also have $\chi_u(D_0(f))+\chi_u(D_1(f))+\chi_u(D_2(f))=0.$ Hence, we have $\chi_u(D_j(f))=0$,\; $\chi_u(D_{j+1}(f))=3^{\frac{n-1}{2}}$,\; and $\chi_u(D_{j+2}(f))=-3^{\frac{n-1}{2}}$. By Equation (\ref{wts}), we have $$wt(c_{u}) \in \{23^{r-2},2\left(3^{r-2}+3^{\frac{n-3}{2}}\right),2\left(3^{r-2}+23^{\frac{n-3}{2}}\right)\}.$$
		\item  $u \notin {B_{-}(f)}^{\perp} $ and $u \in B_{-}(f^*).$\\
From Proposition \ref{dual} we have $0=\chi_u(D_0(f))+\chi_u(D_1(f))\epsilon_{3}+\chi_u(D_2(f))\epsilon_{3}^2$. Then, we have $\chi_u(D_0(f))=\chi_u(D_1(f))=\chi_u(D_2(f)).$ From Fact \ref{bal} we also have $\chi_u(D_0(f))+\chi_u(D_1(f))+\chi_u(D_2(f))=0.$ Hence, we get $\chi_u(D_0(f))=\chi_u(D_1(f))=\chi_u(D_2(f))=0$. By Equation (\ref{wts-}), we have $$wt (c_{u})=2\left(3^{r-2}+3^{\frac{n-3}{2}}\right).$$
	\end{itemize}

Then, the code $\mathcal {C}_{D_{j_0+1}(f)}$ whose codewords $c_{u}$ are defined by Equation (\ref{defCode}) is a three weight ternary linear code with parameters $[3^{r-1}+3^{\frac{n-1}{2}},r,23^{r-2}]_{3}$. The codeword $c_{u}$ has zero-weight if $u \in  \big{(}B_{-}(f)\big{)}^{\perp}$ i.e., $u \in \Kernel(\theta)$. For $u \notin  \big{(}B_{-}(f)\big{)}^{\perp}$ the non-zero weight codewords are as follows.
	
	$$ wt (c_{u})=\left\{ \begin{array}{lll}
	23^{r-2}& \mbox{if $u \in B_{+}(f^*)$ and $f(u)=j_0$};\\
	2\left(3^{r-2}+3^{\frac{n-3}{2}}\right) & \mbox{if $u \in B_{+}(f^*)$ and $f(u)=j_0+1$ or  $u \in B_{-}(f^*)$ };\\
	2\left(3^{r-2}+23^{\frac{n-3}{2}}\right) & \mbox{if $u \in B_{+}(f^*)$ and $f(u)=j_0+2$ }.
	\end{array}
	\right. $$

For $u \notin  \big{(}B_{-}(f)\big{)}^{\perp}$, let us now evaluate the weight distribution of  $\sC_{D_{j_0+1}(f)}$.
\begin{itemize}
\item $u \in B_{+}(f^*)$ and $f(u)=j_0$ \\
 By Proposition \ref{dualsize},  $ |B_{+}(f^*)|=3^r$. 
Since $f$ is of type\;$(-)$, from Equation \eqref{gauss} and Propositions \ref{f*0}, \ref{dual}, we have   $$-i3^{\frac{n}{2}}\epsilon_{3}^{j_0}=-3^{\frac{n-1}{2}}\sum_{j\in\mathbb{F}_3^{\star}}\left( \frac{j}{3} \right)\epsilon_3^{j_0+j}=\sum_{\alpha \in B_{+}(f^*)} \epsilon_{3}^{f(x)}.$$ By orthogonality relations of character sums we have $|C_{j_0}(f^*)|=3^{r-1}$ and $ |C_{j+j_0}(f^*)|=3^{r-1}-\left( \frac{j}{3} \right)3^{\frac{n-1}{2}})$ for all $j\neq 0 \in F_3$. From Lemma \ref{sonn},  the restriction of $f$ into the subset $ \big{(}B_{-}(f)\big{)}^{\perp}$ is the constant $j_0$. Dividing $3^{r-1}-3^{n-r}$ by $3^{n-r}$ we get  $E_{w_1}=3^{2r-n-1}-1$, where $w_1=23^{r-2}$.
\item $u \in B_{+}(f^*)$ and $f(u)=j_0+1$ or  $u \in B_{-}(f^*)$\\
From the arguments above we have $E_{w_2}=3^{r}-23^{2r-n-1}-3^{r-\frac{n+1}{2}}$, where $w_2=23^{r-2}+23^{\frac{n-3}{2}}$.
\item $u \in B_{+}(f^*)$ and $f(u)=j_0+2$ \\
From the arguments above we have $E_{w_3}=3^{2r-n-1}+3^{r-\frac{n+1}{2}}$, where $w_3=23^{r-2}+43^{\frac{n-3}{2}}$.
\end{itemize}
Hence, the weight distribution of $\mathcal {C}_{D_{j_0+1}(f)}$ is as in Table \ref{tabloo4}.

	\begin{table}[!ht] 
		\begin{center}
			\begin{tabular}{|c|c| }
				\hline
				Hamming weight $a$  &  Multiplicity $E_a $  \\ 
				\hline  \hline  
				0  & 1 \\
				\hline
				$23^{r-2}$  &  $3^{2r-n-1}-1$ \\
				\hline
				$	2\left(3^{r-2}+3^{\frac{n-3}{2}}\right)$ &$ 3^{r}-23^{2r-n-1}-3^{r-\frac{n+1}{2}} $\\
				\hline
				$2\left(3^{r-2}+23^{\frac{n-3}{2}}\right)$  &$3^{2r-n-1}+3^{r-\frac{n+1}{2}}$\\
				\hline
			\end{tabular}
		\end{center}	 \caption{ \label{tabloo4}   The weight distribution of  $\mathcal {C}_{D_{j_0+1}(f)}$ when $n$ is odd.}	
	\end{table}

\begin{rem}
We note that for $f$ satisfying the Condition \ref{odd-}, the linear code $\mathcal {C}_{D_{j_0+1}(f)}$ and the one constructed by Theorem \ref{odd} are equivalent. One can see this by the following arguments: Let $n$ be odd and $f:\mathbb{F}_{3}^{n}\rightarrow\mathbb{F}_{3}$ be a non-weakly regular dual-bent function with satisfying the conditions of Theorem \ref{odd} . Let $g(x)=-f(x)$. Then, $\hat{g}(-\alpha)=\sum_{x \in \mathbb{F}_{3}^n} \epsilon_{3}^{-f(x)+\alpha \cdot x}=\sigma_{-1}(\hat{f}(\alpha)).$ For $\alpha \in B_{+}(f)$ we have $\hat{f}(\alpha)= i3^{\frac{n}{2}} \epsilon_{3}^{\alpha}.$  Since $\sigma_{-1}$ is the conjugation automorphism, then we have $\sigma_{-1}(\hat{f}(\alpha))=- i3^{\frac{n}{2}} \epsilon_{3}^{-f^*(\alpha)}$ which implies $-\alpha \in B_{-}(g).$ Moreover, $f(0)=j_0$ implies $g(0)=-j_0.$ Since $g^*(x)=-f^*(-x)$ we have $C_{j_0+2}(f)=D_{-j_0+1}(g)$. Therefore, $g$ satisfies the Condition \ref{odd-}. If $f$ satisfies the Condition \ref{odd-} then by similar arguments  $g$ satisfies the conditions of Theorem \ref{odd}. Hence, there is a one-to-one correspondence between the set of functions satisfying the conditions of Theorem \ref{odd} and Condition \ref{odd-}. 
\end{rem}

\par Let us evaluate the parameters of $\mathcal {C}_{D_{j_0+1}(F)}$ by MAGMA program for the following non-weakly regular bent functions which satisfy the Condition \ref{odd-}.
\begin{example} \label{ex5}
	Let $F:\mathbb{F}_{3}^{5}\simeq\mathbb{F}_{3}^{3}\times\mathbb{F}_{3}\times\mathbb{F}_{3}\rightarrow\mathbb{F}_{3}$ be a non-weakly regular bent function defined by Equation (\ref{gmmf}), where $F(x_1,x_2,x_3,x_4,x_5)=f^{(x_5)}(x_1,x_2,x_3)+x_4x_5 $. Let $f^{(x_5)}:\mathbb{F}_{3}^{3}\rightarrow\mathbb{F}_{3}$ be weakly regular bent of type ($-$) for $x_5=0$ and of type ($+$) for $x_5=1,$ and $x_5=2,$ where $f^{(0)}(x_1,x_2,x_3)=x_1^2+x_2^2 + x_3^2$, $f^{(1)}(x_1,x_2,x_3)=f^{(2)}(x_1,x_2,x_3)=x_1^2+2x_2^2 +x_3^2$. Then, by employing the Langrange interpolation, we have $F(x_1,x_2,x_3,x_4,x_5)=(-1)\big{(}f^{(0)}(x_1,x_2,x_3)(x_5-1)(x_5-2)+(f^{(1)}(x_1,x_2,x_3)+x_4)(x_5)(x_5-2)
	+(f^{(2)}(x_1,x_2,x_3)+2x_4)(x_5)(x_5-1) \big{)}=x_2^2x_5^2+x_1^2+x_2^2+x_3^2+x_4x_5$ is non-weakly regular of type ($-$), where $W^{-}(F)=\{0\}$ and $B_{-}(F)=\mathbb{F}_{3}^{3}\times\{0\}\times\mathbb{F}_{3}$. By \cite[Proposition 2]{Cesmelioglu2013}, we have ${f^{(0)}}^*(x_1,x_2,x_3)=2x_1^2+2x_2^2 + 2x_3^2$,\; ${f^{(1)}}^*(x_1,x_2,x_3,x_4,x_5)={f^{(2)}}^*(x_1,x_2,x_3,x_4,x_5)=2x_1^2+x_2^2 + 2x_3^2$. Moreover, from Equation \eqref{gmmfdual} we know that $F^{*}$ belongs to GMMF class. Then, by employing the Langrange interpolation to the functions ${f^{(0)}}^*$, ${f^{(1)}}^*$, and ${f^{(2)}}^*$, we get $F^{*}(x)=2x_2^2x_4^2+2x_1^2+2x_2^2+2x_3^2+2x_4x_5$.
	
	\begin{itemize}

			\item $m=3,\;s=1$,\; $F(x)=F(-x)$ and  $j_0=0$ ;
		\item $r=m+s+\Dim(W^{-}(F))=3+1+0=4$;
		\item The set $\mathcal {C}_{D_1(F)}$ is a three-weight ternary linear code with parameters $[36,4,18]_3$,  weight enumerator $1+ 8y^{18} + 60y^{24}+ 12y^{30} $.
		
	\end{itemize}
\end{example}
\begin{example} \label{ex6} 
	
	Let $F:\mathbb{F}_{3}^{7}\simeq\mathbb{F}_{3}^{5}\times\mathbb{F}_{3}\times\mathbb{F}_{3}\rightarrow\mathbb{F}_{3}$ be a non-weakly regular bent function defined by Equation (\ref{gmmf}), where $F(x_1,x_2,x_3,x_4,x_5,x_7)=f^{(x_7)}(x_1,x_2,x_3,x_4,x_5)+x_6x_7 $. Let $f^{(x_7)}:\mathbb{F}_{3}^{5}\rightarrow\mathbb{F}_{3}$ be weakly regular bent of type ($-$) for $x_7=0$ and of type ($+$) for $x_7=1,$ and $x_7=2,$ where $f^{(0)}(x_1,x_2,x_3,x_4)=2x_1^2+x_2^2 + x_3^2 + x_4^2+ x_5^2$, $f^{(1)}(x_1,x_2,x_3,x_4)=f^{(2)}(x_1,x_2,x_3,x_4)=x_1^2+x_2^2 + x_3^2 + x_4^2+ x_5^2$. Then, by employing the Langrange interpolation, we have $F(x_1,x_2,x_3,x_4,x_5,x_6,x_7)=(-1)\big{(}f^{(0)}(x_1,x_2,x_3,x_4,x_5)(x_7-1)(x_7-2)+(f^{(1)}(x_1,x_2,x_3,x_4,x_5)+x_6)(x_7)(x_7-2)\\
	+(f^{(2)}(x_1,x_2,x_3,x_4,x_5)+2x_6)(x_7)(x_7-1) \big{)}=2x_1^2x_7^2 + 2x_1^2 + x_2^2 + x_3^2 + x_4^2 + x_7^2+ x_5^2 + x_6x_7+2$ is non-weakly regular of type ($-$), where $W^{-}(F)=\{0\}$ and $B_{-}(F)=\mathbb{F}_{3}^{5}\times\{0\}\times\mathbb{F}_{3}$. By \cite[Proposition 2]{Cesmelioglu2013}, we have ${f^{(0)}}^*(x_1,x_2,x_3,x_4,x_5)=x_1^2+2x_2^2 + 2x_3^2 + 2x_4^2+ 2x_5^2$,\; ${f^{(1)}}^*(x_1,x_2,x_3,x_4,x_5)={f^{(2)}}^*(x_1,x_2,x_3,x_4,x_5)=2x_1^2+2x_2^2 + 2x_3^2 + 2x_4^2 + 2x_5^2$. Moreover, from Equation \eqref{gmmfdual} we know that $F^{*}$ belongs to GMMF class. Then, by employing the Langrange interpolation to the functions ${f^{(0)}}^*$, ${f^{(1)}}^*$, and ${f^{(2)}}^*$, we get $F^{*}(x)=x_1^2x_6^2 + x_1^2 + 2x_2^2 + 2x_3^2 + 2x_4^2 + 2x_6^2+ 2x_5^2 + 2x_6x_7+2$.

	\begin{itemize}
		\item $m=5,\;s=1$,\; $F(x)=F(-x)$ and  $j_0=2$ ;
		\item $r=m+s+\Dim(W^{-}(F))=5+1+0=6$;
		\item The set $\mathcal {C}_{D_0(F)}$ is a three-weight linear code with parameters $[270,6,162]_3$,  weight enumerator $1+ 80y^{162} + 558y^{180} + 90y^{198}$.
		
	\end{itemize}
\end{example}
\begin{rem}
	Since $\lfloor \frac{n}{2}\rfloor+1 \le r \le n$ is a free parameter, for fixed $n$, the size of the defining sets of the constructed codes in this paper can vary, so the other parameters, with respect to the different functions whereas the linear codes from (weakly) regular bent functions based on the second generic construction method in \cite{ding2015class, zhou2016linear, tang} have unique parameters for fixed $n$. Hence, the linear codes constructed in Theorems \ref{even}, \ref{odd}, and \ref{even-} are also distunguished from ones in the literature in this sense.
\end{rem}
In addition to examples above, let us give two sporadic examples of non-weakly regular bent functions. We evaluate the parameters of the corresponding linear codes by MAGMA.  For $g_1$ and $g_2$ we refer to \cite{hel2} and \cite{hel3} respectively.

\begin{example} $g_1:\mathbb{F}_{3^6}\rightarrow\mathbb{F}_{3}$, $g_1(x)=Tr_{6}(\lambda x^{20}+\lambda^{41}x^{92})$ be a non-weakly regular bent of Type ($-$), where $\lambda$ is a primitive element of $\mathbb{F}_{3^6}$.

\begin{itemize}
                     \item  $g_1^{*}(x)$ is bent;
		\item $n=6$,\; $g_1(x)=g_1(-x)$ and  $j_0=0$ ;
		\item $r=4$,\; $B_{-}(g_1)$ is a $4$-dimensional subspace of  $\mathbb{F}_{3^6}$;
		\item The set $\mathcal {C}_{D_2(g_1)}$ is a three-weight linear code with parameters $[36,4,18]_3$,  weight enumerator $1+ 4y^{18} + 72y^{24} + 4y^{36} $.
		
	\end{itemize}
\end{example}

\begin{example}
$g_2:\mathbb{F}_{3^4}\rightarrow\mathbb{F}_{3}$, $g_2(x)=Tr_{4}(w^{10} x^{22}+x^{4})$ be a non-weakly regular bent function of Type ($+$), where $w$ is a primitive element of  $\mathbb{F}_{3^4}$. 

\begin{itemize}

                     \item  $g_2^{*}(x)$ is not bent;

		\item  $n=4$,\; $g_2(x)=g_2(-x)$, and  $j_0=0$;  

                      \item  $r=3$,\; $B_{+}(g_2)$ is a $3$-dimensional subspace of  $\mathbb{F}_{3^4}$;
		
		\item The set $\mathcal {C}_{C_0(g_2)}$ is a three-weight ternary linear code with parameters $[14,3,6]_3$, weight enumerator $1+ 4y^{6} + 18y^{10} + 4y^{12} $.
		
	\end{itemize}

  \end{example}

\begin{rem}
The existence of infinitely many non-weakly regular dual-bent functions $f \notin$ GMMF which satisfy the conditions of at least one of the Theorems \ref{even}, \ref{odd}, or \ref{even-} is an open problem.
\end{rem}
\section{Conclusion}\label{conc}

This paper studies construction of linear codes from non-weakly regular bent functions over finite fields. It should be stated that we used a generic construction method, but the defining sets that we used are new. The main contributions of this paper are as follows. For a given non-weakly regular ternary dual-bent function $f$ with $B_{\pm}(f)$ is a non-degenerate subspace, we construct three weight linear codes by using the pre-image sets of $f^*$ in $B_{\pm}(f)$ as the defining sets of the corresponding codes. Since, the size of the defining sets of the constructed codes are flexible, one can construct several codes with different parameters for a fixed $n$. We also give the weight distribution of the constructed codes. As a future work, for arbitrary odd prime $p$, it should be interesting to construct few weight $p$-ary linear codes from non-weakly regular bent functions based on the second generic construction method.

\end{document}